\documentclass[journal]{IEEEtran}

\usepackage{amsthm}
\usepackage{amssymb}
\usepackage{amsmath}
\usepackage{amssymb}
\usepackage{epsfig}
\usepackage{epsf}
\usepackage{subfigure}
\usepackage{graphicx}
\usepackage{cite}
\usepackage{url}
\usepackage[]{authblk}
\usepackage{dblfloatfix}

\def\BibTeX{{\rm B\kern-.05em{\sc i\kern-.025em b}\kern-.08em
    T\kern-.1667em\lower.7ex\hbox{E}\kern-.125emX}}

\newtheorem{theorem}{Theorem}
\newtheorem{definition}{Definition}

\newtheorem{lemma}{Lemma}

\newtheorem{remark}{Remark}

\begin{document}







\title{Throughput Region of Spatially Correlated Interference Packet Networks}

\author{Alireza~Vahid,~\IEEEmembership{Member,~IEEE,}
Robert~Calderbank,~\IEEEmembership{Fellow,~IEEE}
\thanks{Alireza Vahid is with the Department of Electrical Engineering, University of Colorado Denver, Denver, CO, USA. Email: {\sffamily alireza.vahid@ucdenver.edu}.}
\thanks{Robert Calderbank is with departments of Electrical and Computer Engineering, Mathematics, and Computer Science, Duke University, Durham, NC, USA. Email: {\sffamily robert.calderbank@duke.edu}.}
\thanks{Preliminary results of this work were presented at the 2016 IEEE International Symposium on Information Theory (ISIT)~\cite{vahid2016does}. A shortened version of this paper will appear in ISIT 2018~\cite{vahidARQ}.}
\thanks{Copyright (c) 2017 IEEE. Personal use of this material is permitted. However, permission to use this material for any other purposes must be obtained from the IEEE by sending a request to {\sffamily pubs-permissions@ieee.org}.}
}

\maketitle


\begin{abstract}
In multi-user wireless packet networks interference, typically modeled as packet collision, is the throughput bottleneck. Users become aware of the interference pattern via feedback and use this information for contention resolution and for packet retransmission. Conventional random access protocols interrupt communication to resolve contention which reduces network throughput and increases latency and power consumption. In this work we take a different approach and we develop opportunistic random access protocols rather than pursuing conventional methods. We allow wireless nodes to communicate without interruption and to observe the interference pattern. We then use this interference pattern knowledge and channel statistics to counter the negative impact of interference. We prove the optimality of our protocols using an extremal rank-ratio inequality. An important part of our contributions is the integration of spatial correlation in our assumptions and results. We identify spatial correlation regimes in which inherently outdated feedback becomes as good as idealized instantaneous feedback, and correlation regimes in which feedback does not provide any throughput gain. To better illustrate the results, and as an intermediate step, we characterize the capacity region of finite-field spatially correlated interference channels with delayed channel state information at the transmitters.
\end{abstract}




\section{Introduction}
\label{Section:Introduction}

Recent years have seen a dramatic increase in wireless data traffic thanks to new paradigms such as heterogeneous networks and device-to-device communications. But perhaps more important than the increase in data traffic, is the fact that the character of wireless communications is changing. We are experiencing a transition from a wireless world with human users to one with machines. In this new wireless world we need to accommodate coexistence over the same band of a massive number of communicating machines. However, this task is quite challenging as it is hard to predict when these machines wish to communicate. Moreover, the interference environment constantly changes and at the time of communication, we may not know what this environment looks like. These challenges demand new solutions beyond the conventional medium-access protocols.

Conventional random-access protocols such as ALOHA~\cite{abramson1970aloha} and ALOHA-type mechanisms interrupt communication to resolve contention. This interruption will result in long silent periods in massive machine-type communication (mMTC) which will reduce network throughput and will increase latency. In low-power wide-area networks (LP-WANs), the interruption will result in more retransmission and increased power consumption. Both these scenarios are undesirable. One approach to mitigate this problem is to add more and more signaling in order to allow users to come to an agreement and share the medium, but this approach comes at the cost of further payload reduction.

In this work we take a different approach and we develop {\it open-loop opportunistic communication protocols}. We allow wireless nodes to communicate without interruption and to observe the interference pattern. We then use this interference pattern knowledge and channel statistics to counter the negative impact of interference. The result of our approach is better use of resources (\emph{e.g.}, time and energy) and higher network throughput. We consider two problems in this work. More precisely, we consider a wireless packet network with multiple transmitter-receiver pairs and within this network, we focus on two nearby pairs. We develop an opportunistic random-access or communication protocol for this network that updates the status of previously communicated packets based on the observed interference pattern. These packets are then combined and retransmitted efficiently through the medium. We also prove the optimality of our proposed scheme by developing an extremal rank-ratio inequality. 

\begin{figure}[ht]
\centering
\subfigure[]{\includegraphics[width = \columnwidth]{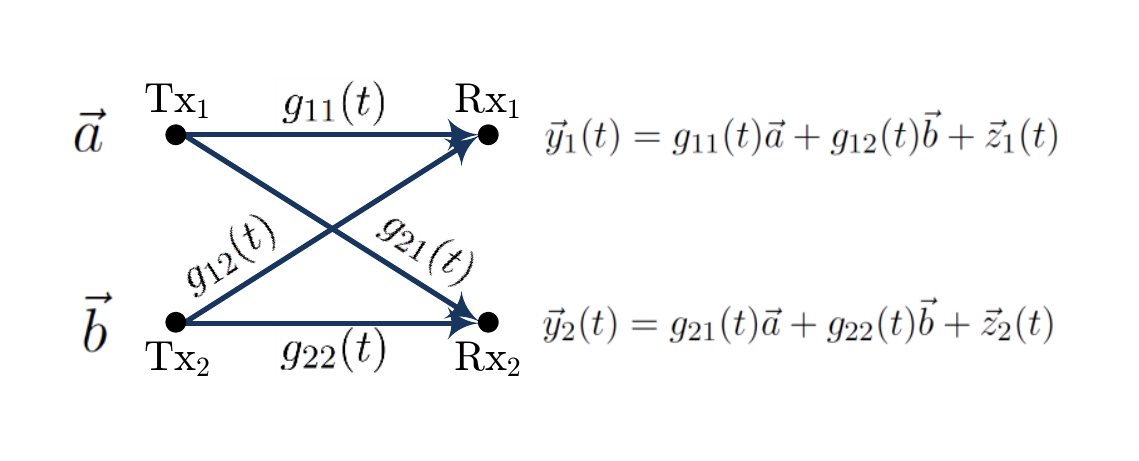}}
\subfigure[]{\includegraphics[width = \columnwidth]{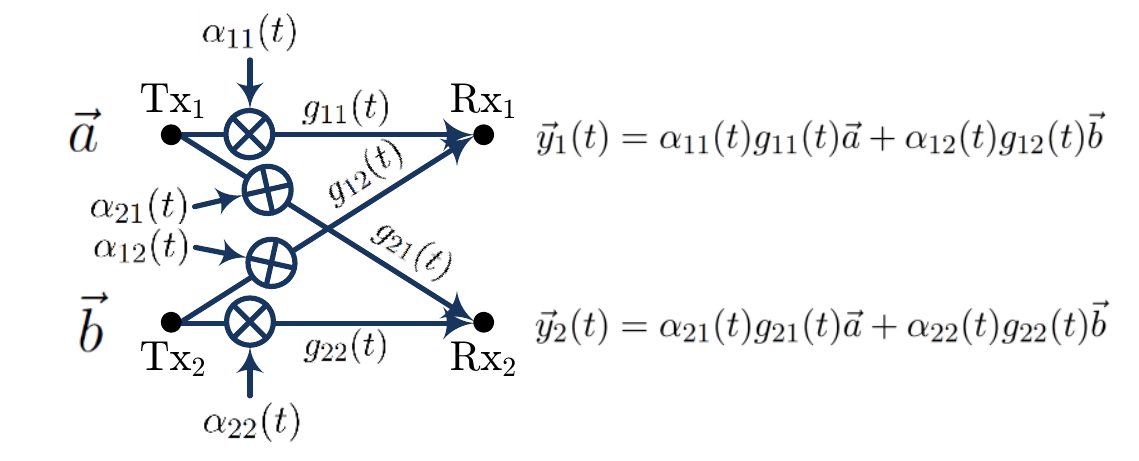}}
\caption{At time $t$, ${\sf Tx}_1$ and ${\sf Tx}_2$ communicate data packets $\vec{a}$ and $\vec{b}$ respectively: (a) $\vec{z}_1(t)$ and $\vec{z}_2(t)$ capture the ambient interference plus noise; (b) $\alpha_{ji}(t)$'s are the shadowing coefficients. \label{Fig:IC-Gaussian}}
\end{figure}

We adopt the physical layer model for wireless packet networks introduced in~\cite{AlirezaInfocom2014} that allows the flexibility of storing the analog received signals at the receivers so that they can be utilized for decoding packets in the future. In this model the impact of interference from nearby users is captured by binary shadowing coefficients as described below. Suppose transmitters ${\sf Tx}_1$ and ${\sf Tx}_2$ communicate data packets $\vec{a}$ and $\vec{b}$ respectively at time $t$ as in Fig.~\ref{Fig:IC-Gaussian}(a). Then receiver ${\sf Rx}_1$ obtains
\begin{align}
\vec{y}_1(t) = g_{11}(t) \vec{a} + g_{12}(t) \vec{b} + \vec{z}_1(t),
\end{align}
where $g_{11}(t)$ and  $g_{12}(t)$ are real-valued channel gains and  $\vec{z}_1(t)$ is the ambient interference plus noise. In the physical layer model of~\cite{AlirezaInfocom2014} this received signal is represented by
\begin{align}
\vec{y}_1(t) = \alpha_{11}(t) g_{11}(t) \vec{a} + \alpha_{12}(t) g_{12}(t) \vec{b},
\end{align}
where the shadowing coefficients $\alpha_{11}(t)$ and $\alpha_{12}(t)$ are in the binary field and depend on the signal-to-interference-plus-noise ratios (SINRs) of different links and determine the interference pattern. This abstraction is depicted in Fig.~\ref{Fig:IC-Gaussian}(b) and is closely related to Erasure Interference Channel (IC) model~\cite{AlirezaBFICDelayed,vahid2016two} which in turn is a generalization of Erasure Broadcast Channel (BC) model~\cite{jolfaei1993new,dana2005capacity,georgiadis2009broadcast,sundararajan2008arq,wang2012capacity,linbursty} to capture interference.

In this work, we assume wireless nodes learn the binary quadruple 
\begin{align}
\label{Eq:Coefficients}
\alpha(t) = \left\{ \alpha_{11}(t), \alpha_{12}(t), \alpha_{21}(t), \alpha_{22}(t) \right\}
\end{align}
with unit delay. This model is commonly referred to as delayed channel state information at the transmitters or delayed CSIT for short. 

\begin{figure}[ht]
\centering
\includegraphics[height = 2 in]{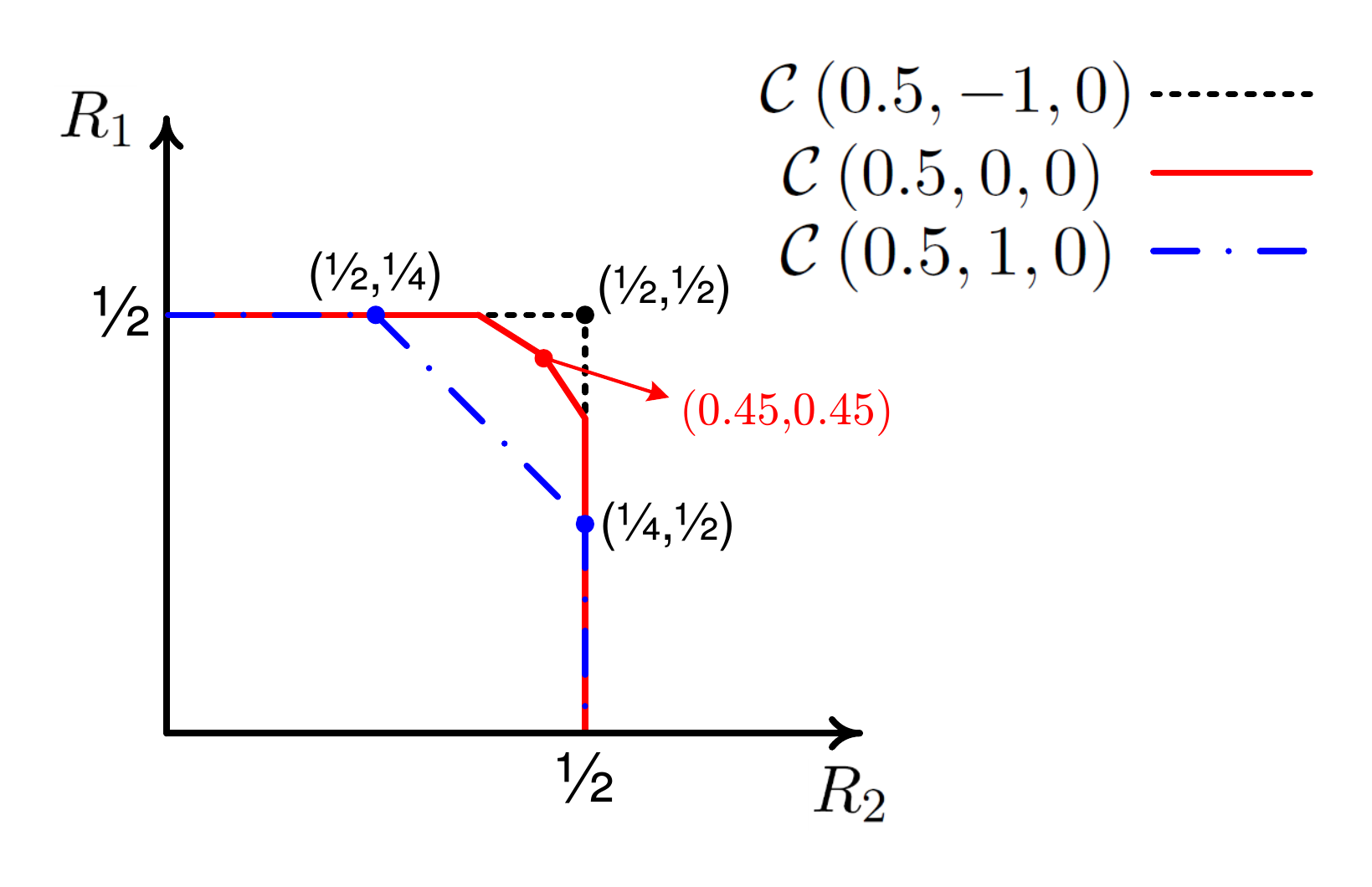}
\caption{Throughput region for $\mathcal{B}(0.5)$ shadowing coefficients and three different correlation coefficient pairs.\label{Fig:CapacityHalf}}
\end{figure}

In most prior work fading coefficients are assumed to be independently and identically distributed across {\it time} and {\it space}. There are recent results on temporal correlation in wireless networks but there is very little known about spatial correlation. Here, we incorporate spatial correlation into our setting. We assume that there is a certain spatial correlation between these binary shadowing coefficients and that this knowledge is available to all nodes as side information. We show that spatial correlation can greatly affect the throughput region of wireless packet networks: spatial correlation on the one hand can take away any potential gain of delayed interference pattern knowledge and on the other hand, it can help us perform as well as having instantaneous knowledge of interference pattern. To see this dichotomy, we focus on the case in which all shadowing coefficients in (\ref{Eq:Coefficients}) are distributed as Bernoulli $\mathcal{B}(0.5)$ random variables and assume that the coefficients corresponding to the wireless links connected to each receiver are distributed independently. When the two channels connected to each transmitter are fully correlated (\emph{i.e.} correlation coefficient of $1$), the throughput region coincides with the one where transmitters do not have any access to interference pattern knowledge as shown in Fig.~\ref{Fig:CapacityHalf}. At the other extreme point where the two channels connected to each transmitter have a correlation coefficient of $-1$, the gain of delayed interference pattern is accentuated and the throughput region is as good as the throughput region of a network in which a genie informs wireless nodes of the interference pattern before it even happens. 


To derive the outer-bounds, a new set of inequalities is needed. We first derive the outer-bounds on the capacity region of spatially correlated finite-field interference channels. To do so, we introduce an extremal entropy inequality for finite-field channels. This inequality captures the ability of each transmitter to favor one receiver over the other in terms of available entropy when channels are correlated.  We use genie-aided arguments and apply our extremal entropy inequality to derive the outer-bounds. We then focus on deriving the outer-bounds on the throughput region of spatially correlated packet networks. Similar to the finite-field problem, we develop an extremal rank-ratio inequality. We show that this inequality is tight and can be achieved in the context of packet networks, and we highlight its similarities to the extremal entropy inequality. We again use genie-aided arguments and apply our extremal rank-ratio inequality to derive the outer-bounds. The slope of these outer-bounds is determined by the correlation coefficient at each transmitter. On the receiver side, correlation coefficients determine the size of the available subspace at each receiver. In other words, spatial correlation at the transmitter side defines the shape of the throughput region while spatial correlation at the receiver side determines its size. 

To achieve the outer-bounds, we develop an opportunistic communication that updates the status of each transmitted packet into three queues: 1) delivered packets, 2) packets that arrived (and potentially interfered) at both receivers; and 3)  packets that arrived (and potentially interfered) at the unintended receivers. The protocol then combines the packets in the latter two queues taking into account the spatial correlation structure of the network. These combined packets will be transmitted efficiently through the medium and upon completion of the communication protocol, each receiver will have sufficient number of equations to recover its intended packets.  

As mentioned above, there is a large body of work on wireless networks with delayed knowledge of the interference pattern or the channel state information. This delayed knowledge was used in~\cite{jolfaei1993new} to create transmitted signals that are simultaneously useful for multiple users in a broadcast channel. These ideas were then extended to different wireless networks, including erasure broadcast channels~\cite{georgiadis2009broadcast,wang2012capacity} leading to determination of the Degrees of Freedom (DoF) region of multiple-input single-output (MISO) Gaussian BCs~\cite{maddah2012completely}. Approximate capacity of MISO BCs was presented in~\cite{vahid2016approximate}. The DoF region of multi-antenna multi-user Gaussian interference channels and X channels has been also considered~\cite{abdoli2013degrees,GhasemiX1,vaze2012degrees,Jafar_Retrospective,tandon2013degrees,nazer2012ergodic,lashgari2014linear,abdoli2014layered,vahid2015informational,vahidXISIT}, as well as the capacity region of Erasure ICs with delayed knowledge~\cite{AlirezaBFICDelayed,vahid2016two}.   

Other results consider more realistic settings in which wireless links exhibit some sort of correlation. Temporal correlation has been studied in several papers~\cite{kobayashi2012degrees,yang2013degrees,de2016optimal,gou2012optimal,yi2014degrees,chen2014symmetric} with the idea being that correlation across time allows transmitters to better estimate what will happen next, and to adjust their transmission strategies accordingly. There are also some results that consider spatial correlation in wireless networks. Point-to-point multiple-input multiple-output (MIMO) with spatial correlation are studied in several papers, \emph{e.g.},~\cite{gesbert2002outdoor,goldsmith2003capacity,love2008overview}. In the context of broadcast channels, there are several results that propose transmission strategies that take into account spatial correlation of wireless links~\cite{raghavan2013statistical,choi2013limited,dai2015transmit,clerckx2015space,zhang2017spatially}. In BCs all messages are available to one node, and thus decision making is centralized. However, in interference networks this is no longer the case, and the decentralized nature of the problem introduces new challenges. 

In the context of interference channels, a transmission strategy was developed in~\cite{nosrat2011mimo} for MIMO ICs with {\it instantaneous} channel knowledge and spatial correlation. In~\cite{raghavan2010statistical}, the performance of treating-interference-as-noise is studied, and in~\cite{filippou2013optimal}, a beamforming strategy is developed for spatially correlated wireless networks. However, a fundamental understanding of the capacity region of spatially correlated interference channels is missing and is the focus of this work. We characterize the capacity region of spatially correlated finite-field interference channels and the throughput region of spatially correlated interference packet networks.

It is also worth noting a separate but related line of work on spatial correlation. In~\cite{duan2016state,duan2013state,ghasemi2014state}, authors consider a class of interference channels in which the received signals at different receivers are corrupted by some ambient signals. Such results do not consider correlated channel realization, rather they assume correlation among the ambient signals. 

The rest of the paper is organized as follows. In Section~\ref{Section:Problem} we formulate our problem and justify the physical layer model. In Section~\ref{Section:Main} we present our main results and provide some insights. Sections~\ref{Section:Converse} and~\ref{Section:Achievability} are dedicated to the proof of the main results. We describe connections to two related directions in Section~\ref{Section:Discussion}, and Section~\ref{Section:Conclusion} concludes the paper.


\section{Problem Formulation}
\label{Section:Problem}

We quantify the impact of spatial correlation on the capacity region of finite-field interference channels and on the throughput region of wireless packet networks. In this section we define each model, and we introduce the notations we use in this paper.

\subsection{Finite-Field Interference Channels}

In the finite-field model, the channel gain from transmitter ${\sf Tx}_i$ to receiver ${\sf Rx}_j$ at time $t$ is in the binary field and is denoted by $\alpha_{ji}(t)$, $i,j \in \{1,2\}$. These channel gains are distributed as Bernoulli random variables with parameter $p$, \emph{i.e.} $\alpha_{ji}(t) \overset{d}\sim \mathcal{B}(p)$. The channel state information at time instant $t$ is denoted by
\begin{align}
\alpha(t) \overset{\triangle}= \left \{ \alpha_{11}(t), \alpha_{12}(t), \alpha_{21}(t), \alpha_{22}(t) \right \}.
\end{align} 
To simplify the equations in this paper, we set 
\begin{align}
\label{Eq:Notationq}
q \overset{\triangle}=1-p, \qquad \text{and} \qquad \bar{i} = 3 - i.
\end{align}

\begin{figure}[ht]
\centering
\includegraphics[height = 3.5cm]{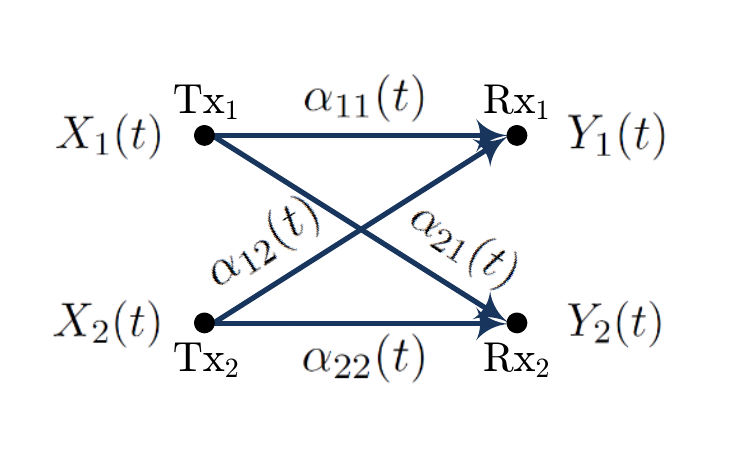}
\caption{Two-user Finite-Field Interference Channel.\label{Fig:detIC}}
\end{figure}

We assume that the transmitters become aware of the CSI, $\alpha(t)$, with unit delay. Since each receiver only needs to decode its message at the end of the communication block, without loss of generality, we assume that each receiver has instantaneous knowledge of the channel state information.

This channel is illustrated in Fig.~\ref{Fig:detIC} where the input-output relation at time $t$ is given by
\begin{equation} 
Y_i(t) = \alpha_{ii}(t) X_i(t) \oplus \alpha_{i\bar{i}}(t) X_{\bar{i}}(t), \quad i = 1, 2,
\end{equation}
where all algebraic operations are in $\mathbb{F}_2$, $\bar{i} = 3 - i$ as defined in (\ref{Eq:Notationq}), $X_i(t) \in \{ 0, 1\}$ is the transmit signal of ${\sf Tx}_i$ at time $t$, and $Y_i(t) \in \{ 0, 1\}$ is the observation of ${\sf Rx}_i$ at time $t$. We note that the observation of ${\sf Rx}_i$ and the channel state information form the received signal of ${\sf Rx}_i$. 

\begin{remark}
In the finite-field model we focus on the binary coefficients and the reason is as follows. In Section~\ref{Section:Achievability} we show that in order to achieve the capacity, what matters is whether or not transmit signal $X_i(t)$ contributes to the receiver observation $Y_j(t)$. In other words, once the transmit signal shows up at the receiver, the channel coefficient becomes irrelevant. This means even if we consider a larger finite field, all that matters is whether each channel gain is zero or not which is essentially a binary random variable. See Remark~\ref{Remark:BinaryField} at the end of Section~\ref{Section:Example1}  for more details.
\end{remark}

\subsection{Interference Packet Networks}

In this subsection we connect the physical-layer model defined above to a network-layer model. In particular, we consider a wireless packet network in which multiple transmitter-receiver pairs wish to communicate with each other and within this setup, we focus on two nearby transmitter-receiver pairs, namely ${\sf Tx}_1$-${\sf Rx}_1$ and ${\sf Tx}_2$-${\sf Rx}_2$ as depicted in Fig.~\ref{Fig:WLAN}. 

\begin{figure}[ht]
\centering
\includegraphics[height = 2.5 in]{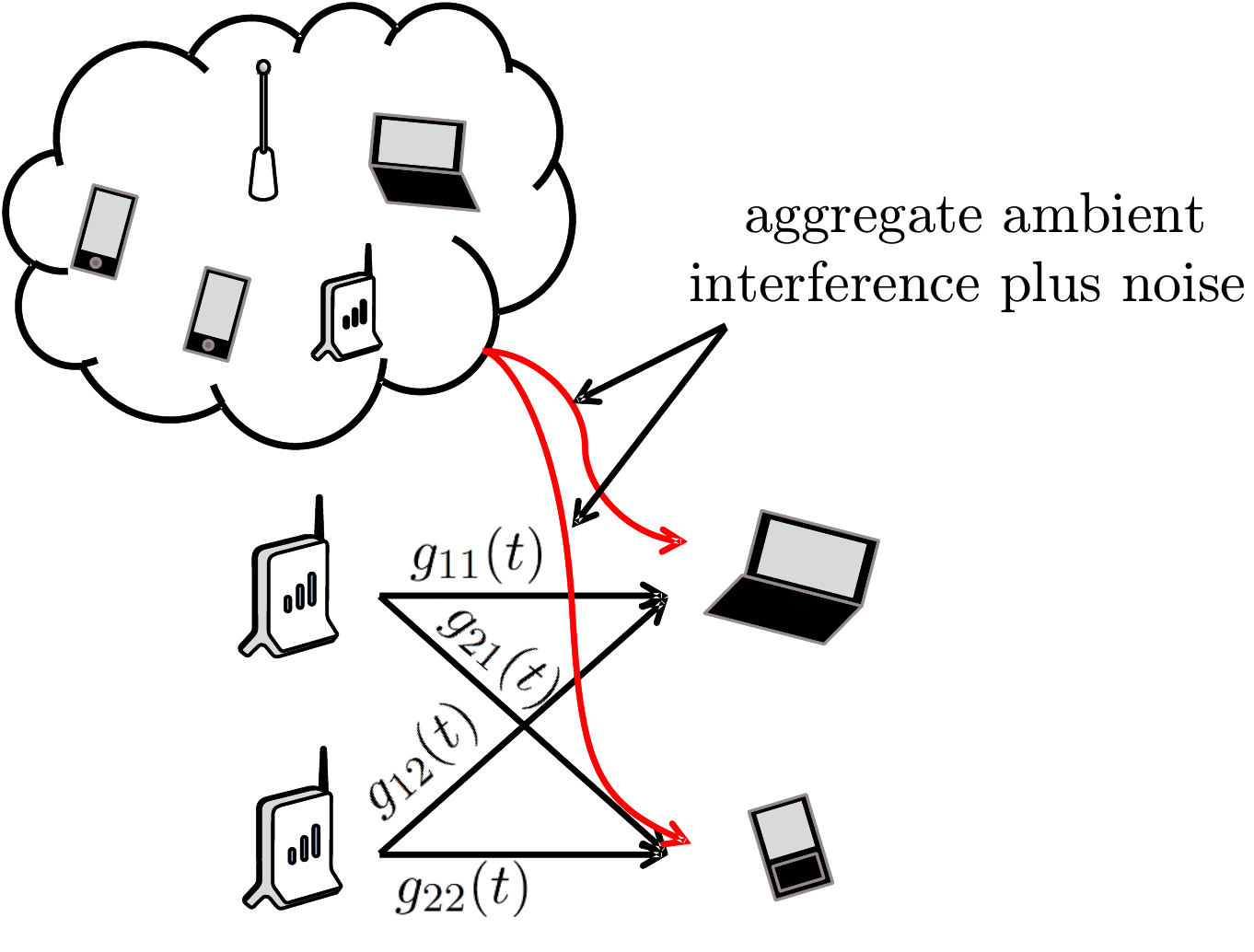}
\caption{A wireless packet network in which multiple transmitter-receiver pairs communicate with each other. We focus on two nearby pairs.\label{Fig:WLAN}}
\end{figure}

We adopt the abstraction for wireless packet networks introduced in~\cite{AlirezaInfocom2014}. In this model transmitter ${\sf Tx}_1$ has $m_1$ packets with corresponding physical layer codewords of length $\tau$ denoted by $\vec{a}_1, \vec{a}_2, \ldots, \vec{a}_{m_1}$ and wishes to communicate them to receiver ${\sf Rx}_1$. Similarly, transmitter ${\sf Tx}_2$ has $m_2$ packets denoted by $\vec{b}_1, \vec{b}_2, \ldots, \vec{b}_{m_2}$ for receiver ${\sf Rx}_2$. It is assumed that the mapping from the packets to their corresponding physical layer codewords is fixed (\emph{e.g.}, LDPC codes, Reed-Solomon codes, etc) and if a codeword is received with a signal-to-interference-plus-noise ratio above $\gamma$, the receiver is able to decode its packet. The signal-to-interference-plus-noise ratio of link $ji$ at ${\sf Rx}_j$, $i,j \in \{ 1, 2\}$, is defined as:
\begin{align}
\mathrm{SINR}_{ji} \overset{\triangle}= 10 \log_{10} \left( \frac{P|g_{ji}|^2}{\mathbb{E}\left[ \vec{z}_j^\top(t) \vec{z}_j(t) \right] + P |g_{j\bar{i}}|^2} \right),
\end{align}
where $P$ is the average transmit power constraint and $\vec{z}_j(t)$ is the additive noise plus the ambient interference (from potential nearby transmitters except for ${\sf Tx}_i$). Furthermore, we define the signal-to-noise ratio (SNR) of link $ji$ at ${\sf Rx}_j$, $i,j \in \{ 1, 2\}$, as:
\begin{align}
\mathrm{SNR}_{ji} \overset{\triangle}= 10 \log_{10} \left( \frac{P|g_{ji}|^2}{\mathbb{E}\left[ \vec{z}_j^\top(t) \vec{z}_j(t) \right]} \right).
\end{align}

\begin{remark}
As mentioned above, a packet can be successfully decoded if the corresponding codeword is received with a signal-to-interference-plus-noise ratio above $\gamma$. In other words, if the SINR is above $\gamma$, then the packet can be decoded by treating-interference-as-noise.  This might sound like we intend to treat interference as noise in our system. However, as we will see in Section~\ref{Section:Achievability}, for the majority of the communication block, the SINR threshold is not met and treating-interference-as-noise would be suboptimal. Instead of discarding the codewords received below the threshold, we keep the ones received with an SNR above $\gamma$ and we treat the received signal as a linear combination of the collided packets. We devise an interference alignment technique that capitalizes on these equations and achieves the outer-bound. The proposed communication scheme significantly outperforms treating-interference-as-noise as discussed in Section~\ref{Section:Main}.  
\end{remark}

\begin{figure}[ht]
\centering
\subfigure[State~1:~$\mathrm{SINR}_{11} \ge \gamma$]{\includegraphics[width=0.49\columnwidth]{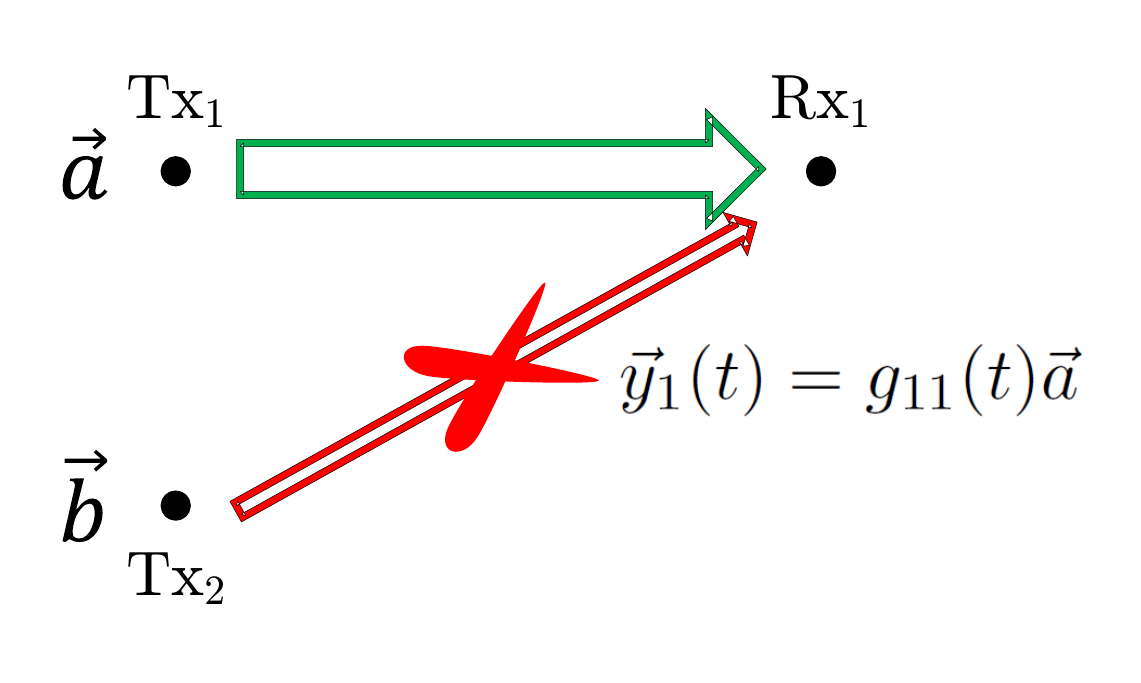}}
\subfigure[State~2:~$\mathrm{SINR}_{12} \ge \gamma$]{\includegraphics[width=0.49\columnwidth]{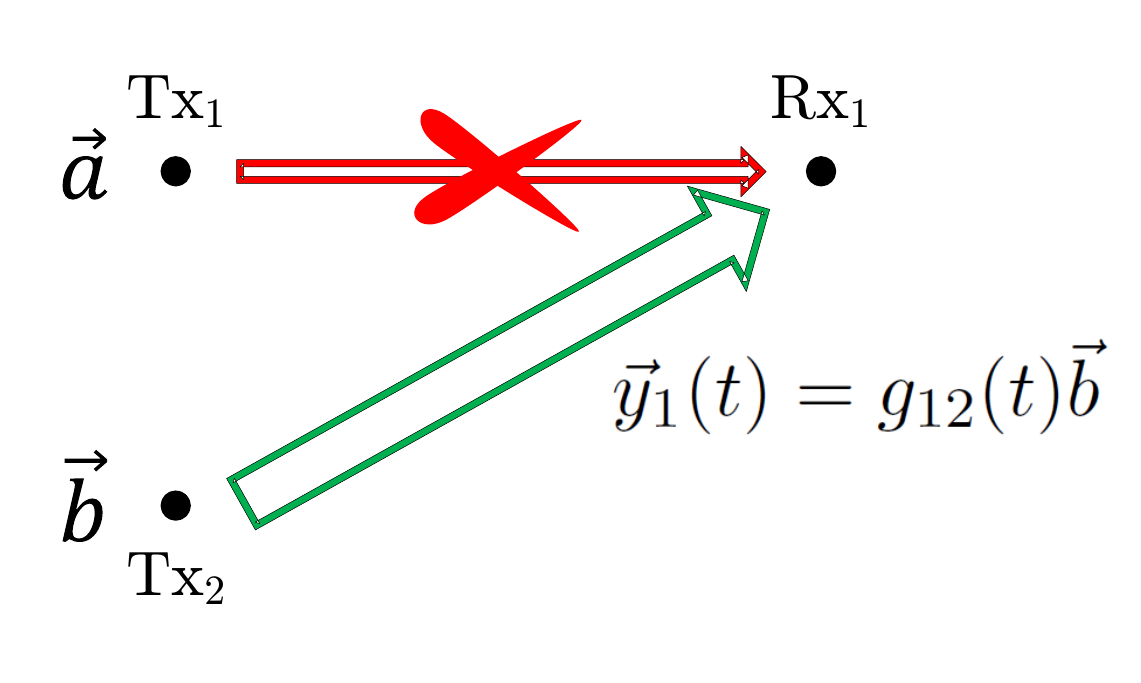}}
\subfigure[State~3:~$\mathrm{SINR}_{1i} < \gamma, \mathrm{SNR}_{1i} \ge \gamma$]{\includegraphics[width=0.49\columnwidth]{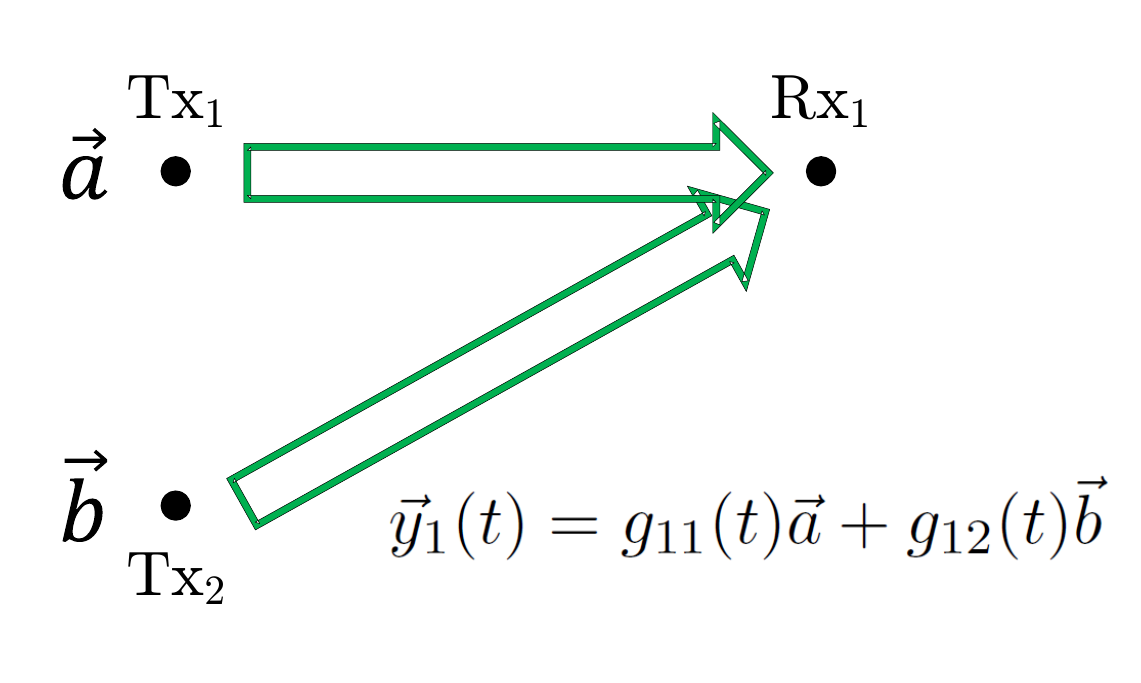}}
\subfigure[State~4]{\includegraphics[width=0.49\columnwidth]{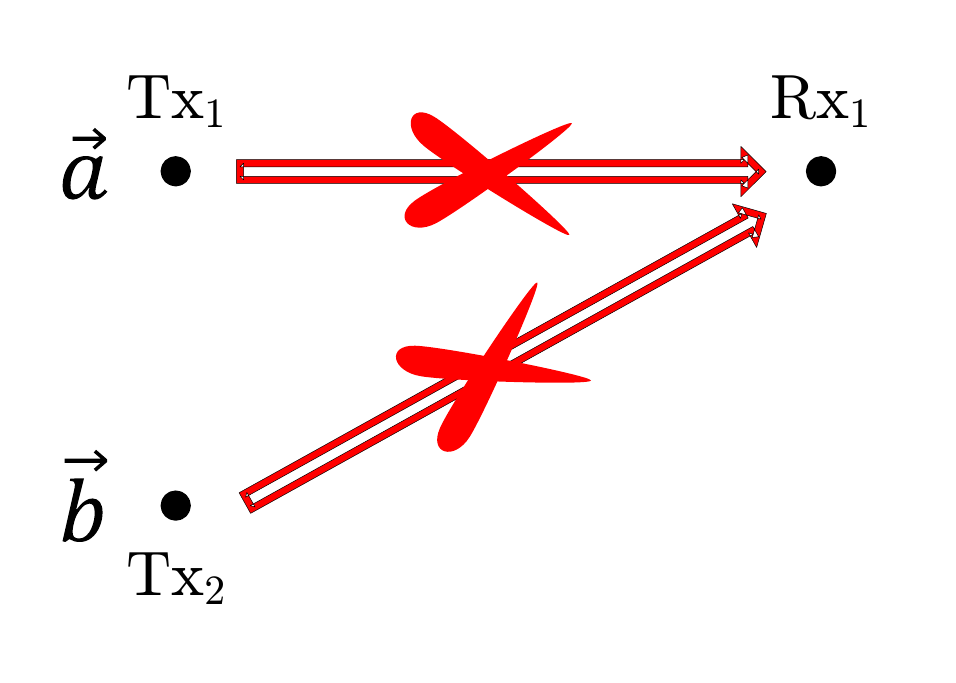}}
\caption{Based on the SINR and SNR values of different links at each time $t$, we have four states.}\label{Fig:OnOff}
\end{figure}

Based on the SINR and the SNR values of different links at each time $t$, we have one of the following states at any of the receivers, say ${\sf Rx}_1$:
\begin{itemize}

\item State~1~($\mathrm{SINR}_{11} \ge \gamma$): In this state the SINR of the desired packet (\emph{i.e.} $\vec{a}$) at ${\sf Rx}_1$ is above the threshold and the packet can be decoded by treating the interference as noise. Fig.~\ref{Fig:OnOff}(a) depicts this scenario. 

\item State~2~($\mathrm{SINR}_{12} \ge \gamma$): Similar to State~1, but in this case the SINR of the interfering packet (\emph{i.e.} $\vec{b}$) at ${\sf Rx}_1$ is above the threshold and ${\sf Rx}_1$ (the unintended receiver in this case) can decode this packet, see Fig.~\ref{Fig:OnOff}(b).

\item State~3~($\mathrm{SINR}_{1i} < \gamma$ but $\mathrm{SNR}_{1i} \ge \gamma$ for $i=1,2$): This state corresponds to the scenario in which the SINRs of both packets ($\vec{a}$ and $\vec{b}$) are below the threshold at ${\sf Rx}_1$ but the individual links are strong (\emph{i.e.} $\mathrm{SNR}_{1i} \ge \gamma$). In this case neither of the packets can be decoded by treating the interference as noise. However, ${\sf Rx}_1$ stores the signal as the weighted linear combination of packets $\vec{a}$ and $\vec{b}$ as depicted in Fig.~\ref{Fig:OnOff}(c).  

\item State~4: In any other scenario, ${\sf Rx}_1$ discards the received signal, see Fig.~\ref{Fig:OnOff}(d).

\end{itemize}

The following abstraction of the physical layer at ${\sf Rx}_1$ captures the four different states:
\begin{align}
\label{eq:alphaij}
\vec{y}_1(t) = \alpha_{11}(t) g_{11}(t) \vec{a} + \alpha_{12}(t) g_{12}(t) \vec{b},
\end{align}
where the shadowing coefficients $\alpha_{ji}(t)$'s are in the binary field, $i,j \in \{ 1, 2\}$. For example in State~3, we have $\left( \alpha_{11}(t) = 1, \alpha_{12}(t) = 1 \right)$ and
\begin{align}
\label{eq:case3}
\vec{y}_1(t) = g_{11}(t) \vec{a} + g_{12}(t) \vec{b}.
\end{align}
Similarly, $\alpha_{21}(t)$ and $\alpha_{22}(t)$ are used for ${\sf Rx}_2$. 

As in the finite-field model, the channel state at time instant $t$ is represented by quadruple 
\begin{align}
\alpha(t) = \left \{ \alpha_{11}(t), \alpha_{12}(t), \alpha_{21}(t), \alpha_{22}(t) \right \}.
\end{align}
We assume that $\alpha_{ji}(t)$ follows a Bernoulli distribution $\mathcal{B}\left( p \right)$ and that at time $t$, each transmitter knows $\alpha^{t-1} = \left( \alpha(\ell) \right)_{\ell = 1}^{t-1}$ and each receiver has access to $\alpha^{t} = \left( \alpha(\ell) \right)_{\ell = 1}^{t}$.  

\subsection{Spatial Correlation}

In general, $\alpha_{ji}(t)$'s are correlated across time and space. As mentioned in the introduction, correlation across time allows us to predict the future and improve the throughput using this prediction. On the other hand, there is little known about spatial correlation and that is the focus of the current paper. Since we study the impact of spatial correlation, in this work we assume that $\alpha_{ji}(t)$'s are distributed independently across time and are drawn from the same joint distribution at each time. To capture spatial correlation, we assume a symmetric setting in which the shadowing coefficients corresponding to the links connected to transmitter ${\sf Tx}_i$ have a correlation coefficient $\rho_{\sf Tx}$, \emph{i.e.}
\begin{align}
\label{Eq:rhoTx}
\rho_{\sf Tx} = \frac{\mathrm{cov}\left( \alpha_{11}(t), \alpha_{21}(t) \right)}{\sigma_{\alpha_{11}(t)}\sigma_{\alpha_{21}(t)}} = \frac{\mathrm{cov}\left( \alpha_{12}(t), \alpha_{22}(t) \right)}{\sigma_{\alpha_{12}(t)}\sigma_{\alpha_{22}(t)}}, 
\end{align}
and the links connected to receiver ${\sf Rx}_j$ have a correlation coefficient $\rho_{\sf Rx}$, \emph{i.e.}
\begin{align}
\rho_{\sf Rx} = \frac{\mathrm{cov}\left( \alpha_{11}(t), \alpha_{12}(t) \right)}{\sigma_{\alpha_{11}(t)}\sigma_{\alpha_{12}(t)}} = \frac{\mathrm{cov}\left( \alpha_{21}(t), \alpha_{22}(t) \right)}{\sigma_{\alpha_{21}(t)}\sigma_{\alpha_{242}(t)}}. 
\end{align}

We note that fixing $-1 \leq \rho_{\sf Tx}, \rho_{\sf Rx} \leq 1$ imposes a feasible set on $p$. More precisely, we require $p$ to be in $\mathcal{S}_{\rho_{\sf Tx}} \cap \mathcal{S}_{\rho_{\sf Rx}}$ where $\mathcal{S}_{\rho_{\sf Tx}}$ and $\mathcal{S}_{\rho_{\sf Rx}}$ are defined below. First, we define
\begin{align}
\label{Eq:Pij}
& p^{\sf Tx}_{k\ell} \overset{\triangle}= \\
& \Pr\left( \alpha_{11}(t) = k, \alpha_{21}(t) = \ell \right) = \Pr\left( \alpha_{22}(t) = k, \alpha_{12}(t) = \ell \right),  \nonumber \\
& p^{\sf Rx}_{k\ell} \overset{\triangle}= \nonumber\\ 
& \Pr\left( \alpha_{11}(t) = k, \alpha_{12}(t) = \ell \right) = \Pr\left( \alpha_{22}(t) = k, \alpha_{21}(t) = \ell \right), \nonumber 
\end{align}
for $k,\ell \in \{ 0, 1 \}$. Then, $\mathcal{S}_{\rho_{\sf Tx}} \subseteq \left[ 0, 1\right]$ is the set of all values for $p$ such that
\begin{equation}
\label{Eq:RangeTx}
\left\{ \begin{array}{ll}
\vspace{1mm} 0 \leq p^{\sf Tx}_{00}, p^{\sf Tx}_{10}, p^{\sf Tx}_{01}, p^{\sf Tx}_{11} \leq 1 &  \\
\vspace{1mm} \sum_{i,j \in \{ 0, 1 \}}p^{\sf Tx}_{ij} = 1 &  \\
\vspace{1mm} p^{\sf Tx}_{10} + p^{\sf Tx}_{11} = p &  \\
\vspace{1mm} p^{\sf Tx}_{01} + p^{\sf Tx}_{11} = p &  \\
p q \rho_{\sf Tx} = p^{\sf Tx}_{11} - p^2 &
\end{array} \right.
\end{equation}
where the first two conditions follow the probability space definition; the third and the fourth equalities hold since
\begin{align}
p & = \Pr \left( \alpha_{11}(t) = 1 \right) = \Pr \left( \alpha_{11}(t) = 1, \alpha_{21}(t) = 0 \right) \nonumber\\
& +  \Pr \left( \alpha_{11}(t) = 1, \alpha_{21}(t) = 1 \right) \overset{(\ref{Eq:Pij})}= p^{\sf Tx}_{10} + p^{\sf Tx}_{11}, \nonumber \\
p & = \Pr \left( \alpha_{21}(t) = 1 \right) = \Pr \left( \alpha_{11}(t) = 1, \alpha_{21}(t) = 1 \right) \nonumber \\
& +  \Pr \left( \alpha_{11}(t) = 0, \alpha_{21}(t) = 1 \right) \overset{(\ref{Eq:Pij})}= p^{\sf Tx}_{01} + p^{\sf Tx}_{11};
\end{align}
the last equality is derived from (\ref{Eq:rhoTx}) as follows:
\begin{align}
\rho_{\sf Tx} & = \frac{\mathrm{cov}\left( \alpha_{11}(t), \alpha_{21}(t) \right)}{\sigma_{\alpha_{11}(t)}\sigma_{\alpha_{21}(t)}} \nonumber\\
& = \frac{\mathbb{E}\left[ \alpha_{11}(t) \alpha_{21}(t) \right] - \mathbb{E}\left[ \alpha_{11}(t) \right] \mathbb{E}\left[ \alpha_{21}(t) \right]}{pq} \nonumber \\ 
& \Rightarrow p q \rho_{\sf Tx} = p^{\sf Tx}_{11} - p^2.
\end{align}
Thus, we get
\begin{equation}
\left\{ \begin{array}{ll}
\vspace{1mm} p^{\sf Tx}_{11} = pq\rho_{\sf Tx}+p^2 &  \\
\vspace{1mm} p^{\sf Tx}_{10} = p - pq\rho_{\sf Tx} - p^2 &  \\
\vspace{1mm} p^{\sf Tx}_{01} = p - pq\rho_{\sf Tx} - p^2 &  \\
p^{\sf Tx}_{00} = 1 - p^{\sf Tx}_{11} - p^{\sf Tx}_{10} - p^{\sf Tx}_{01} &
\end{array} \right.
\end{equation}

Similarly, $\mathcal{S}_{\rho_{\sf Rx}} \subseteq \left[ 0, 1\right]$ is the set of all values for $p$ such that
\begin{equation}
\label{Eq:RangeRx}
\left\{ \begin{array}{ll}
\vspace{1mm} 0 \leq p^{\sf Rx}_{00}, p^{\sf Rx}_{10}, p^{\sf Rx}_{01}, p^{\sf Rx}_{11} \leq 1 &  \\
\vspace{1mm} \sum_{i,j \in \{ 0, 1 \}}p^{\sf Rx}_{ij} = 1 &  \\
\vspace{1mm} p^{\sf Rx}_{10} + p^{\sf Rx}_{11} = p &  \\
\vspace{1mm} p^{\sf Rx}_{01} + p^{\sf Rx}_{11} = p &  \\
p q \rho_{\sf Rx} = p^{\sf Rx}_{11} - p^2 &
\end{array} \right.
\end{equation}
which gives us
\begin{equation}
\left\{ \begin{array}{ll}
\vspace{1mm} p^{\sf Rx}_{11} = pq\rho_{\sf Rx}+p^2 &  \\
\vspace{1mm} p^{\sf Rx}_{10} = p - pq\rho_{\sf Rx} - p^2 &  \\
\vspace{1mm} p^{\sf Rx}_{01} = p - pq\rho_{\sf Rx} - p^2 &  \\
p^{\sf Rx}_{00} = 1 - p^{\sf Rx}_{11} - p^{\sf Rx}_{10} - p^{\sf Rx}_{01} &
\end{array} \right.
\end{equation}
Moreover, we have 
\begin{align}
& p^{\sf Tx}_{11} \geq 0 \overset{\rho_{\sf Tx} \neq 1}\Rightarrow p \geq \frac{-\rho_{\sf Tx}}{1-\rho_{\sf Tx}}, \nonumber \\
& p^{\sf Rx}_{11} + p^{\sf Rx}_{10} + p^{\sf Rx}_{01} \leq 1 \overset{\rho_{\sf Tx} \neq 1}\Rightarrow p \leq \frac{1}{1-\rho_{\sf Tx}}
\end{align}
Thus, $\mathcal{S}_{\rho_{\sf Tx}}$ is given by
\begin{align}
\mathcal{S}_{\rho_{\sf Tx}} \overset{\triangle}= \left[ \max\left\{ 0, \frac{-\rho_{\sf Tx}}{1-\rho_{\sf Tx}} \right\}, \min\left\{ 1, \frac{1}{1-\rho_{\sf Tx}} \right\} \right],
\end{align}
and similarly, $\mathcal{S}_{\rho_{\sf Rx}}$ is given by
\begin{align}
\mathcal{S}_{\rho_{\sf Rx}} \overset{\triangle}= \left[ \max\left\{ 0, \frac{-\rho_{\sf Rx}}{1-\rho_{\sf Rx}} \right\}, \min\left\{ 1, \frac{1}{1-\rho_{\sf Rx}} \right\} \right],
\end{align}
where we set $\mathcal{S}_{\rho_{\sf Tx} = 1}, \mathcal{S}_{\rho_{\sf Rx} = 1} \overset{\triangle}= \left[ 0, 1 \right]$. 

\subsection{Encoding, Decoding, and Probability of Error}

\subsubsection{Finite-Field Model}

For the finite-field interference channel we consider the scenario in which ${\sf Tx}_i$ wishes to communicate message $\hbox{W}_i \in \{ 1,2,\ldots,2^{n R_i}\}$ to ${\sf Rx}_i$ during $n$ channel uses, $i = 1,2$. We assume that the messages and the channel gains are {\it mutually} independent and the messages are chosen uniformly. Let message $\hbox{W}_i$ be encoded as $X_i^n$ at ${\sf Tx}_i$, where at time $t$ we have $X_i[t] = f_i(\hbox{W}_i, G^{t-1})$. Receiver ${\sf Rx}_i$ is only interested in decoding $\hbox{W}_i$, and it will decode the message using the decoding function $\widehat{\hbox{W}}_i = g_i(Y_i^n,G^n)$, and the knowledge of $\rho_{\sf Tx}$ and $\rho_{\sf Rx}$. An error occurs when $\widehat{\hbox{W}}_i \neq \hbox{W}_i$. The average probability of decoding error is given by
\begin{equation}
\label{eq:errorterms}
\lambda_{i,n} = \mathbb{E}[P[\widehat{\hbox{W}}_i \neq \hbox{W}_i]], \hspace{5mm} i = 1, 2,
\end{equation}
where the expectation is taken with respect to the random choice of the transmitted messages $\hbox{W}_1$ and $\hbox{W}_2$. 

A rate-tuple $(R_1,R_2)$ is said to be achievable if there exist encoding and decoding functions at the transmitters and the receivers respectively, such that the decoding error probabilities $\lambda_{1,n},\lambda_{2,n}$ go to zero as $n$ goes to infinity. 

\begin{definition}
\label{Def:CapDef}
The capacity region of the two-user finite-field interference channel with parameter $p \in \mathcal{S}_{\rho}$, correlation coefficient $\rho_{\sf Tx}$ at the transmitters, and correlation coefficient $\rho_{\sf Rx}$ at the receivers, is denoted by $\mathcal{C}\left( p, \rho_{\sf Tx}, \rho_{\sf Rx} \right)$, and is the closure of all achievable rate-tuples defined above. 
\end{definition}

\subsubsection{Packet Networks}

Consider the scenario in which ${\sf Tx}_i$ wishes to communicate $m_i$ packets\footnote{We note that in this paper we assume transmitters start with a number of packets to communicate rather than stochastic arrival of packets.} to ${\sf Rx}_i$ during $n$ uses of the channel, $i = 1,2$. We assume that the packets and the channel gains are {\it mutually} independent. Receiver ${\sf Rx}_i$ is only interested in packets from ${\sf Tx}_i$, and it will recover (decode) them using the received signal $\vec{y}_i^n$, the knowledge of the channel state information, and the knowledge of $\rho_{\sf Tx}$ and $\rho_{\sf Rx}$. 

At each time instant, transmitter $i$ creates a linear combination of the $m_i$ packets it has for receiver $i$ by choosing a precoding vector $\vec{v}_i(t) \in \mathbb{R}^{1 \times m_i}$, $i=1,2$. Transmit signals at time $t$ at ${\sf Tx}_1$ and ${\sf Tx}_2$ are given by $\vec{v}_1(t) \mathbf{A}$ and $\vec{v}_2(t) \mathbf{B}$ respectively, where $\mathbf{A} = \left[ \vec{a}_1, \vec{a}_2, \ldots, \vec{a}_{m_1} \right]^{\top}$, and $\mathbf{B} = \left[ \vec{b}_1, \vec{b}_2, \ldots, \vec{b}_{m_2} \right]^{\top}$. We impose the following constraints on $\vec{v}_1(t)$ and $\vec{v}_2(t)$ to satisfy the power constraint (\emph{i.e.} an average transmit power of $P$) at the transmitters:
\begin{align}
||\vec{v}_1(t)||, ||\vec{v}_2(t)|| \leq 1,
\end{align}
where $||.||$ denotes the Euclidean norm. Since transmitters learn the CSI, $\alpha(t)$, with unit delay, $\vec{v}_i(t)$ is only a function of $\alpha^{t-1}$ and the correlation coefficients $\rho_{\sf Tx}$ and $\rho_{\sf Rx}$. The received signal of receiver $i$ at time $t$ can be represented by
\begin{align}
\vec{y}_i(t) = \alpha_{i1}(t) g_{i1}(t) \vec{v}_1(t) \mathbf{A} + \alpha_{i2}(t) g_{i2}(t) \vec{v}_2(t) \mathbf{B}.
\end{align}
We denote the overall precoding matrix of transmitter $i$ by $\mathbf{V}_i^n \in \mathbb{R}^{n \times m_i}$, where the $t^{\mathrm{th}}$ row of $\mathbf{V}_i^n$ is $\vec{v}_i(t)$. Furthermore, let $\mathbf{G}_{ij}^n$ be an $n \times n$ diagonal matrix where the $t^{\mathrm{th}}$ diagonal element is $\alpha_{ij}(t) g_{ij}(t)$, $i,j = 1,2$. Thus, we can write the output at receiver $i$ as
\begin{align}
\label{eq:PHYLayer}
\vec{y}_i^n &\overset{\triangle}= \left( \vec{y}_i(1), \ldots, \vec{y}_i(n) \right)^{\sf T} \nonumber \\
& =  \mathbf{G}_{i1}^n \mathbf{V}_1^n \mathbf{A} + \mathbf{G}_{i2}^n \mathbf{V}_2^n \mathbf{B}, \quad i = 1,2.
\end{align}

We denote the interference subspace at receiver $i$ by $\mathcal{I}_i$ and is given by
\begin{align}
\mathcal{I}_i \overset{\triangle}= \mathrm{colspan}\left( \mathbf{G}_{i\bar{i}} \mathbf{V}_{\bar{i}} \right), \quad i = 1,2,
\end{align}
where $\mathrm{colspan}(.)$ of a matrix represents the subspace spanned by its column vectors. Let $\mathcal{I}_i^c$ denote the subspace orthogonal to $\mathcal{I}_i$. Then in order for decoding to be successful at receiver $i$, it should be able to create $m_i$ linearly independent equations that are solely in terms of its intended packets. Mathematically speaking, this means that the image of $\mathrm{colspan}\left( \mathbf{G}_{ii}^n \mathbf{V}_{i}^n \right)$ on $\mathcal{I}_i^c$ should have the same dimension as  $\mathrm{colspan}\left( \mathbf{V}_{i}^n \right)$ itself. More precisely, we require
\begin{align}
\label{eq:decodability}
\mathrm{dim} & \left( \mathrm{Proj}_{\mathcal{I}_i^c}~\mathrm{colspan}\left( \mathbf{G}_{ii}^n \mathbf{V}_{i}^n \right) \right) \nonumber\\ 
& = \mathrm{dim}\left( \mathrm{colspan}\left( \mathbf{V}_{i}^n \right) \right) = m_i,  \quad i = 1,2.
\end{align}

We say that a throughput tuple of $\left( R_1,R_2 \right) = \left( m_1/n, m_2/n \right)$ is achievable, if there exists a choice of $\mathbf{V}_1^n$ and $\mathbf{V}_2^n$, such that (\ref{eq:decodability}) is satisfied for $i=1,2$ with probability $1$. 

\begin{definition}
\label{Def:ThroughDef}
The throughput region of the two-user interference packet network with parameter $p \in \mathcal{S}_{\rho}$, correlation coefficient $\rho_{\sf Tx}$ at the transmitters, and correlation coefficient $\rho_{\sf Rx}$ at the receivers, is denoted by $\mathcal{T}\left( p, \rho_{\sf Tx}, \rho_{\sf Rx} \right)$, and is the closure of all achievable throughput tuples $\left( R_1,R_2 \right)$ defined above. We note that the unit for $R_i$ is packets per channel use.
\end{definition}

\begin{remark}
\label{Remark:Family}
We would like to point out a subtle point regarding our problem formulation here. Fixing $p, \rho_{\sf Tx},$ and $\rho_{\sf Rx}$ does not specify a unique channel but rather a family of channels as we only fix pairwise correlation coefficients at each node. For instance, we did not specify the correlation coefficient between $\alpha_{11}(t)$ and $\alpha_{22}(t)$, and any correlation coefficient between these two random variables consistent with the specified coefficients is acceptable. For this family of channels, the capacity (throughput) region is identical and given in Theorem~\ref{THM:CapacityCorrelated}. For more discussion on Bernoulli random variables with known pairwise correlations, we refer the reader to~\cite{wainwright2008graphical}. 
\end{remark}


\section{Statement of the Main Results}
\label{Section:Main}

The following theorem establishes the capacity region of the two-user finite-field spatially correlated interference channel and the throughput region of spatially correlated interference packet networks with two transmitter-receiver pairs.

The capacity region for the finite-field model follows the typical information-theoretic formulation and the outer-bounds hold for \emph{all} input distributions. The throughput region of the spatially correlated interference packet networks, however, is derived under the assumption that the encoding at the transmitters is limited to linear functions. In this regard, the throughput region of the spatially correlated interference packet networks is similar to the linear degrees-of-freedom region of wireless networks~\cite{lashgari2014linear,kao2017linear,yang2017linear}.   


\begin{theorem}
\label{THM:CapacityCorrelated}
For the two-user finite-field spatially correlated interference channel and the spatially correlated interference packet network with two transmitter-receiver pairs as described in Section~\ref{Section:Problem} and for $p \in \mathcal{S}_{\rho_{\sf Tx}} \cap \mathcal{S}_{\rho_{\sf Rx}}$, we have
\begin{align}
\label{Eq:Capacity}
& \mathcal{C}\left( p, \rho_{\sf Tx}, \rho_{\sf Rx} \right) = \mathcal{T}\left( p, \rho_{\sf Tx}, \rho_{\sf Rx} \right) = \\
& \left\{ \begin{array}{ll}
\vspace{1mm} 0 \leq R_i \leq p, & \\
R_i + \beta\left( p, \rho_{\sf Tx} \right) R_{\bar{i}} \leq \beta\left( p, \rho_{\sf Tx} \right) \left( 1 - p^{\sf Rx}_{00} \right). &
\end{array} \right\} \nonumber 
\end{align}
for $i = 1,2$, where the capacity and the throughput regions are defined in Definitions~\ref{Def:CapDef}~and~\ref{Def:ThroughDef}, respectively, and\footnote{We point out that for all values of $p \in [ 0, 1]$ and $\rho_{\sf Tx} \in [ -1, 1]$, we have $\beta\left( p, \rho_{\sf Tx} \right) \geq 1$.} 
\begin{align}
\label{Eq:Beta}
\beta\left( p, \rho_{\sf Tx} \right) &= 1 + \left( 1 - \rho_{\sf Tx} \right) \left( 1 - p \right), \nonumber \\
p^{\sf Rx}_{00} &= 1+p^2+pq\rho_{\sf Rx}-2p.
\end{align}
\end{theorem}

The converse proof of Theorem~\ref{THM:CapacityCorrelated} for the finite-field interference channel relies on an extremal entropy inequality for correlated channels that we present in Section~\ref{Section:ConverseFinite}. The slope of the outer-bounds (\emph{i.e.} $\beta\left( p, \rho_{\sf Tx} \right)$ in the theorem) is determined by this inequality and depends on spatial correlation at the transmitters. Using the extremal inequality and genie-aided arguments, we obtain the outer-bound. For the spatially correlated interference packet networks, we modify the extremal entropy inequality to obtain an extremal rank-ratio inequality for correlated channels that we present in Section~\ref{Section:ConversePacket}. We then derive the outer-bounds following similar steps to the converse proof of the finite-field model. 

The achievability proof is very similar for the two models and thus, we only present one communication protocol focusing on packet networks in Section~\ref{Section:Achievability}. The proposed communication protocol has multiple phases of communications and after each phase, transmitters use the delayed interference pattern knowledge to update the status of the previously communicated packets. The goal is to retransmit linear combinations of packets in a way to help receivers decode their corresponding packets faster than the case in which individual packets are retransmitted. We also highlight the key differences between this protocol and the one suitable for independent channel gains.

\begin{figure}[ht]
\centering
\includegraphics[width = \columnwidth]{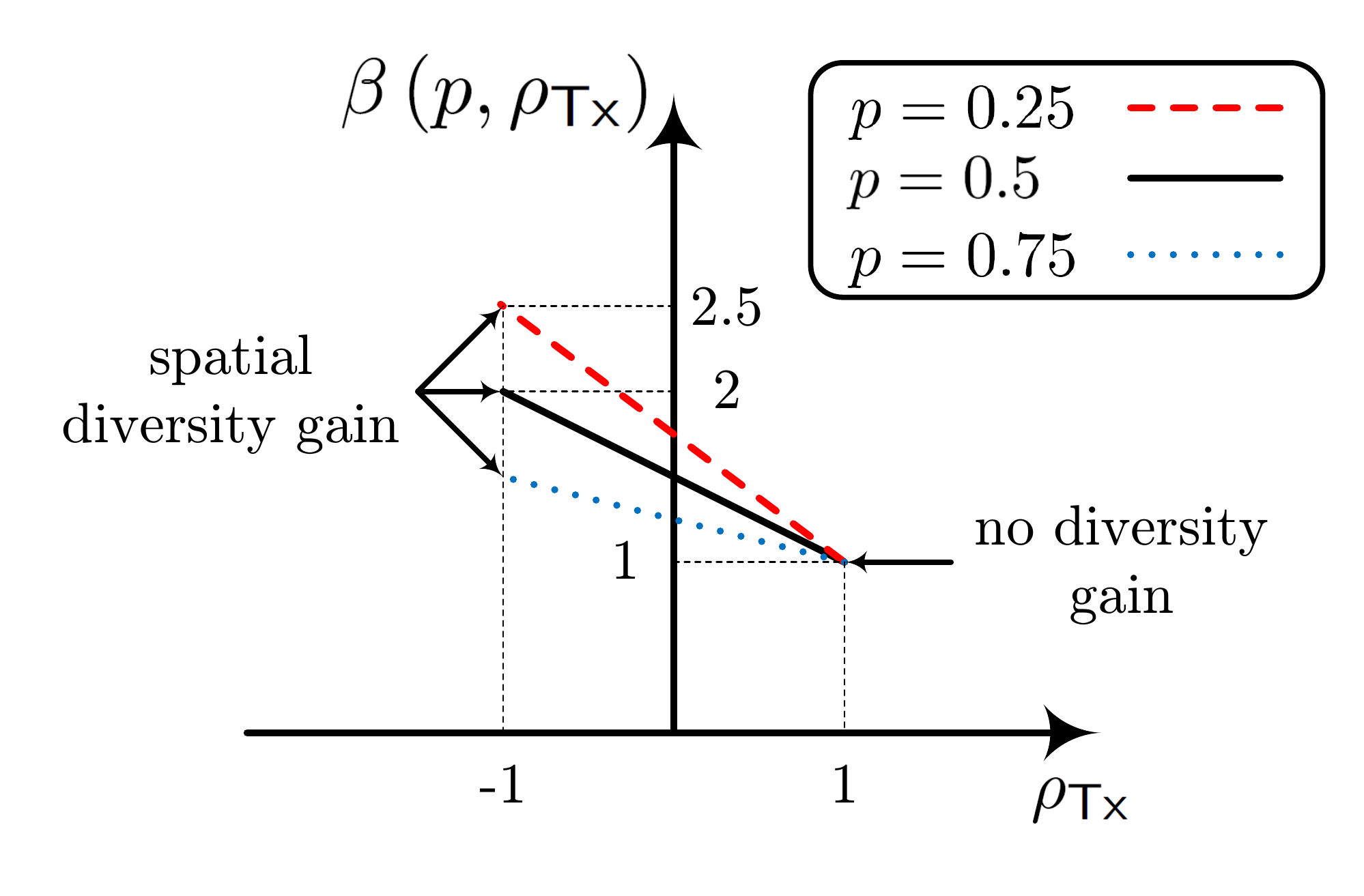}
\caption{Negative values of $\rho_{\sf Tx}$ increases $\beta\left( p, \rho_{\sf Tx} \right)$ and the ability of transmitters to perform interference alignment.\label{Fig:DiversityGain}}
\end{figure}

Before presenting the proofs, we provide further insights and interpretations of Theorem~\ref{THM:CapacityCorrelated}. We start by discussing the importance of $\beta\left( p, \rho_{\sf Tx} \right)$. The value of $\beta\left( p, \rho_{\sf Tx} \right)$ determines the ability of each transmitter to differentiate receivers in terms of the entropy (in the finite-field model) or the number of linearly independent equations (in the packet networks) it can provide to them: $\beta\left( p, \rho_{\sf Tx} \right) = 1$ means the receivers are identical from the perspective of each transmitter and $\beta\left( p, \rho_{\sf Tx} \right) > 1$ provides {\it spatial diversity} meaning that each transmitter can favor one receiver over the other. Fig.~\ref{Fig:DiversityGain} illustrates $\beta\left( p, \rho_{\sf Tx} \right)$ as a function of $\rho_{\sf Tx}$ for $p = 3/4, 1/2, 1/4$. We observe that negative values of $\rho_{\sf Tx}$ increase $\beta\left( p, \rho_{\sf Tx} \right)$ and this observation is consistent with the fact that negative $\rho_{\sf Tx}$ results in higher network throughput as it increases the ability of transmitters to perform interference alignment.

\begin{figure}[ht]
\centering
\includegraphics[width = \columnwidth]{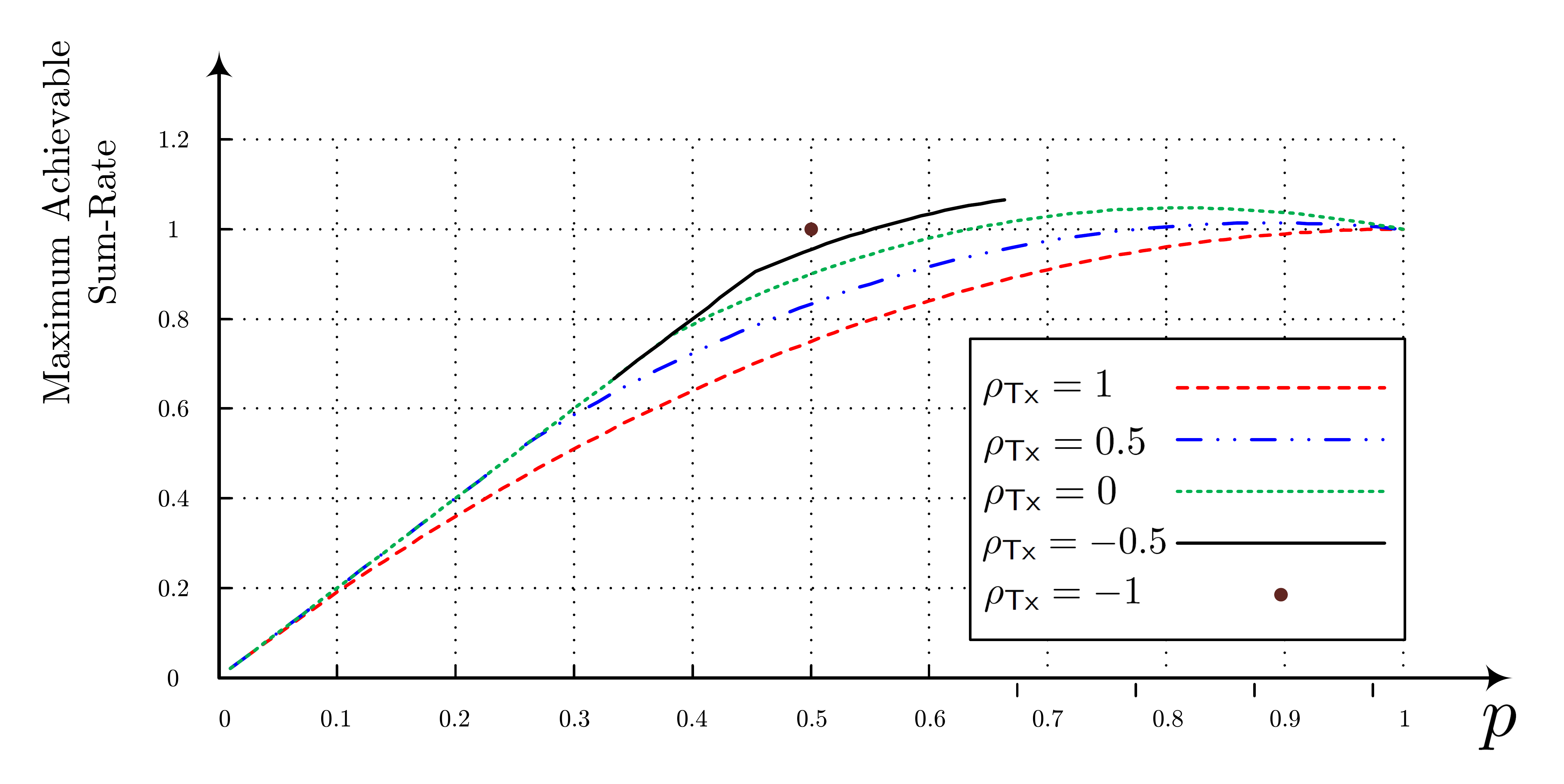}
\caption{Maximum achievable sum-rate for $\rho_{\sf Tx} \in \{ -1, -0.5, 0, 0.5, 1 \}$, $\rho_{\sf Rx} = 0$, and $p \in \mathcal{S}_{\rho_{\sf Tx}} \cap \mathcal{S}_{\rho_{\sf Rx}}$.\label{Fig:SumCap}}
\end{figure}

We then focus on the impact of correlation at the transmitters on the capacity (throughput) region. Note that since $\alpha_{ji}(t)$'s are distributed as $\mathcal{B}(p)$, the maximum throughput we can expect in this network is $2p$. Fig.~\ref{Fig:SumCap} depicts the maximum achievable sum-rate for $\rho_{\sf Tx} \in \{ -1, -0.5, 0, 0.5, 1 \}$, $\rho_{\sf Rx} = 0$ (\emph{i.e.} independent links at the receivers), and $p \in \mathcal{S}_{\rho_{\sf Tx}}$. For a fixed value of $p$, as $\rho_{\sf Tx}$ moves from $+1$ to $-1$, the maximum achievable sum-rate improves. In fact, for $p=0.5$ as discussed in the introduction and shown in Fig.~\ref{Fig:CapacityHalf}, $\mathcal{C}\left( 0.5, 1, 0\right)$ coincides with the one where transmitters do not have any access to interference pattern~\cite{AlirezaInfocom2014}. On the other hand, $\mathcal{C}\left( 0.5, -1, 0\right)$ includes $\left( R_1, R_2 \right) = \left( 0.5, 0.5 \right)$ which implies that the capacity (throughput) region coincides with the capacity (throughput) region of a network in which a genie informs wireless nodes of the interference pattern before it even happens. Intuitively, this is due to the fact that with fully correlated channels, each transmitter cannot distinguish between the two receivers and as a result, it is not able to perform interference cancellation or interference alignment. However, with negative correlation, a transmitter's power to favor one receiver over the other improves (in terms of the number of independent linear combinations received). This in turn enables the transmitters to perform interference alignment and interference cancellation more efficiently. When all channels are distributed independently and identically across time and space, the capacity is given by  $\mathcal{C}\left( 0.5, 0, 0\right)$ as shown in Fig.~\ref{Fig:CapacityHalf} and recovers the results of~\cite{AlirezaInfocom2014}. 

\begin{figure}[ht]
\centering
\includegraphics[height = 5cm]{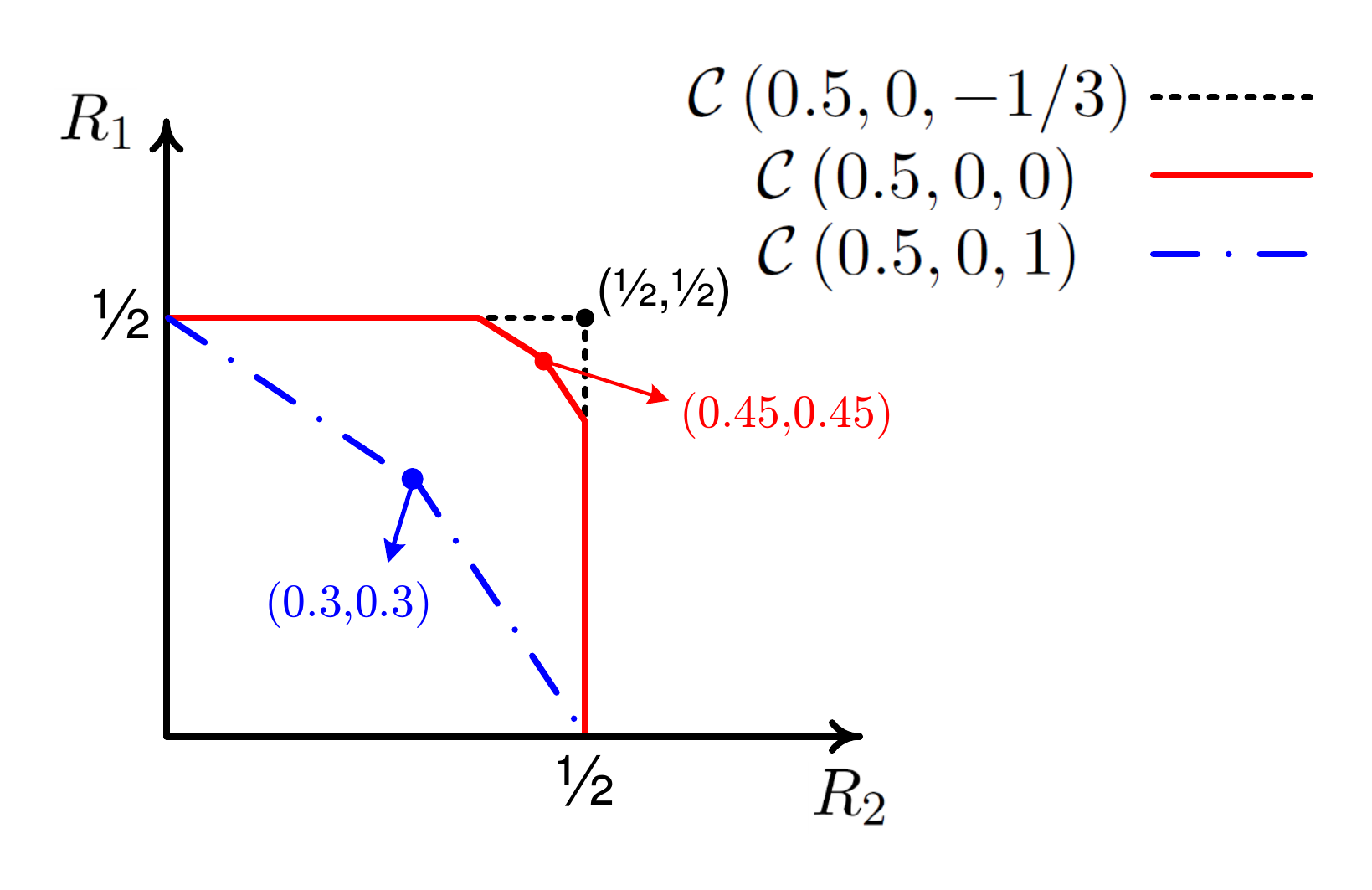}
\caption{Capacity (or throughput) region for $p =0.5$, $\rho_{\sf Tx} = 0$, and $\rho_{\sf Rx} \in \{ -1/3, 0, 1 \}$.\label{Fig:region-receiver}}
\end{figure}

To study the impact of spatial correlation at the receivers on the capacity (throughput) region, we consider $\rho_{\sf Tx} = 0$ (\emph{i.e.} independent links at the transmitters), $\rho_{\sf Rx} \in \{ -1/3, 0, 1 \}$,  and $p = 0.5$. As $\rho_{\sf Rx}$ moves from $+1$ to $-1$, the maximum achievable sum-rate improves as shown Fig.~\ref{Fig:region-receiver}. For $\rho_{\sf Rx} \in \left[ -1, -1/3 \right]$, $\mathcal{C}\left( 0.5, 0, \rho_{\sf Rx} \right)$ includes $\left( R_1, R_2 \right) = \left( 0.5, 0.5 \right)$ which implies that the capacity region coincides with that of instantaneous knowledge. Intuitively, negative spatial correlation at the receivers separates the signal subspace from the interference subspace which results in higher network throughput, see Section~\ref{Section:Achievability} for more details. 

\begin{figure}[ht]
\centering
\includegraphics[width = \columnwidth]{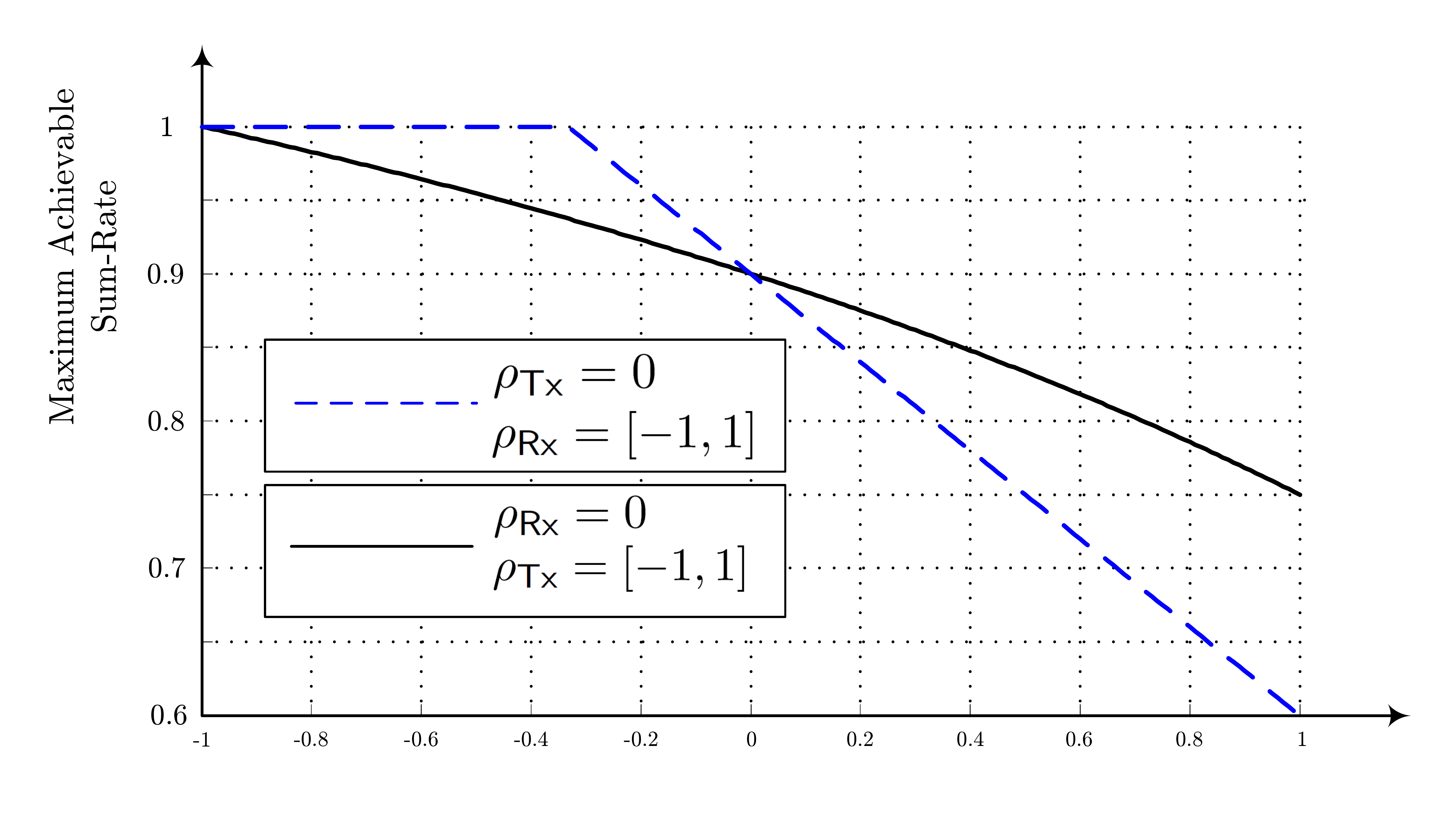}
\caption{The impact of spatial correlation at the transmitter side or at the receiver side on the maximum achievable sum-rate for $p =0.5$. For $\rho_{\sf Tx} = 0$, the X axis represents $\rho_{\sf Rx}$ from $-1$ to $1$, and for $\rho_{\sf Rx} = 0$ the X axis represents $\rho_{\sf Tx}$ from $-1$ to $1$\label{Fig:SumCapBoth}}
\end{figure}

To better visualize the impact of spatial correlation at the transmitters and at the receivers on the capacity (throughput) region, we plot the maximum achievable sum-rate in Fig.~\ref{Fig:SumCapBoth} for two different scenarios with $p=0.5$. First, we set $\rho_{\sf Tx}$ to be zero (dashed blue line), \emph{i.e.} independent links at the transmitters. In this case we see that the impact of $\rho_{\sf Rx}$ is linear on the maximum achievable sum-rate and it saturates at $1$. Then, we set $\rho_{\sf Rx}$ to be zero (solid black line), \emph{i.e.} independent links at the receivers. As mentioned before, $\rho_{\sf Tx}$ determines the slope of the outer-bounds and thus defines the shape of the throughput region, and $\rho_{\sf Rx}$ defines the size of the throughput region. These arguments will be made mathematically precise in the following sections.


\section{Converse Proof of Theorem~\ref{THM:CapacityCorrelated}: Two Extremal Inequalities for Spatially Correlated Channels}
\label{Section:Converse}

The individual rates are limited by the capacity of a point-to-point erasure channel for which the derivation is well-known and thus omitted. The main focus of this section is the derivation of the other bounds. We first derive these bounds for the finite-field model and then, we focus on packet networks. 

\subsection{Finite-Field Model: An Extremal Entropy Inequality}
\label{Section:ConverseFinite}

We first present an extremal entropy inequality tailored to spatially correlated channels.
\begin{lemma}[Extremal Entropy Inequality for Spatially Correlated Channels]
\label{Lemma:EntropyLeakage}
For the two-user finite-field interference channel with delayed CSIT and correlated links as described in Section~\ref{Section:Problem} and for $p \in \mathcal{S}_{\rho_{\sf Tx}} \cap \mathcal{S}_{\rho_{\sf Rx}}$ and $p \neq 0$, we have
\begin{align}
H\left( Y_2^n | W_2, \alpha^n \right) \geq \frac{1}{\beta} H\left( Y_1^n | W_2, \alpha^n \right),
\end{align}
where
\begin{align}
\beta\left( p, \rho_{\sf Tx} \right) = 1 + \left( 1 - \rho_{\sf Tx} \right) \left( 1 - p \right).
\end{align}
\end{lemma}

\begin{proof} 
For time $t$ where $1 \leq t \leq n$, we have
\begin{align}
& H\left( Y_2(t) | Y_2^{t-1}, W_2, \alpha^t \right) \nonumber \\
& \overset{(a)}= H\left( Y_2(t) | Y_2^{t-1}, X_2^t, W_2, \alpha^t \right) \nonumber \\
& = H\left( \alpha_{21}(t)X_1(t) | Y_2^{t-1}, X_2^t, W_2, \alpha^t \right) \nonumber \\
& \overset{(b)}= H\left( \alpha_{21}(t)X_1(t) | Y_2^{t-1}, W_2, \alpha^t \right) \nonumber \\
& \overset{(c)}= p H\left( X_1(t) | Y_2^{t-1}, \alpha_{21}(t) = 1, W_2, \alpha^{t-1} \right) \nonumber \\
& \overset{(d)}= p H\left( X_1(t) | Y_2^{t-1}, W_2, \alpha^t \right) \nonumber \\
& \overset{(e)}\geq p H\left( X_1(t) | Y_1^{t-1},Y_2^{t-1}, W_2, \alpha^t \right) \nonumber \\
& \overset{(f)}= \frac{p}{2p-pq\rho_{\sf Tx}-p^2} H\left( Y_1(t), Y_2(t) | Y_1^{t-1},Y_2^{t-1}, W_2, \alpha^t \right) \nonumber \\
& = \frac{1}{1 + (1 - \rho_{\sf Tx}) (1 - p)} H\left( Y_1(t), Y_2(t) | Y_1^{t-1},Y_2^{t-1}, W_2, \alpha^t \right) \nonumber \\
& \overset{(g)}= \frac{1}{\beta\left( p, \rho_{\sf Tx} \right)} H\left( Y_1(t), Y_2(t) | Y_1^{t-1},Y_2^{t-1}, W_2, \alpha^t \right),
\end{align}
where $(a)$ holds since $X_2^t$ is a function of $\left( W_2, \alpha^t \right)$ and the correlation coefficients; $(b)$ is true since given $\left( W_2, \alpha^t \right)$, $\alpha_{21}(t)X_1(t)$ and $X_2^t$ are independent; $(c)$ follows the fact that $\alpha_{21}(t)$ is in the binary field and $\Pr\left( \alpha_{21}(t) = 1 \right) = p$; $(d)$ holds since $X(t)$ is independent of the channel realization at time $t$; $(e)$ follows from the fact that conditioning reduces entropy; $(f)$ holds since the probability that at least one of the channel gains connected to ${\sf Tx}_1$ is equal to $1$ is one minus the probability that both links are zero, and is given by
\begin{align}
1 - p^{\sf Tx}_{00} = 2p-pq\rho_{\sf Tx}-p^2;
\end{align}
$(g)$ follows the definition of $\beta\left( p, \rho_{\sf Tx} \right)$ given in (\ref{Eq:Beta}) and mentioned above. Therefore, we have 
\begin{align}
& \sum_{t=1}^n{H\left( Y_2(t) | Y_2^{t-1}, W_2, \alpha^t \right)} \\
& \geq \frac{1}{\beta\left( p, \rho_{\sf Tx} \right)} \sum_{t=1}^n{H\left( Y_1(t), Y_2(t) | Y_1^{t-1}, Y_2^{t-1}, W_2, \alpha^t \right)}, \nonumber
\end{align}
and since the transmit signal at time $t$ are independent from future channel realizations, we get
\begin{align}
& \sum_{t=1}^n{H\left( Y_2(t) | Y_2^{t-1}, W_2, \alpha^n \right)} \\
& \geq \frac{1}{\beta\left( p, \rho_{\sf Tx} \right)} \sum_{t=1}^n{H\left( Y_1(t), Y_2(t) | Y_1^{t-1}, Y_2^{t-1}, W_2, \alpha^n \right)}, \nonumber
\end{align}
from which we have
\begin{align}
H\left( Y_2^n | W_2, \alpha^n \right) & \geq \frac{1}{\beta\left( p, \rho_{\sf Tx} \right)} H\left( Y_1^n, Y_2^n | W_2, \alpha^n \right) \nonumber \\
& \geq \frac{1}{\beta\left( p, \rho_{\sf Tx} \right)} H\left( Y_1^n | W_2, \alpha^n \right).
\end{align}
\end{proof}

Using Lemma~\ref{Lemma:EntropyLeakage}, we have
\begin{align}
n &\left( R_1 + \beta\left( p, \rho_{\sf Tx} \right) R_2 \right) \nonumber \\
& = H(W_1) + \beta\left( p, \rho_{\sf Tx} \right) H(W_2) \nonumber \\
& = H(W_1|W_2, \alpha^n) + \beta\left( p, \rho_{\sf Tx} \right) H(W_2|\alpha^n) \nonumber \\
& \overset{(\mathrm{Fano})}\leq I(W_1;Y_1^n|W_2,\alpha^n) + \beta\left( p, \rho_{\sf Tx} \right) I(W_2;Y_2^n|\alpha^n) + n \epsilon_n \nonumber \\
& = H(Y_1^n|W_2,\alpha^n) - \underbrace{H(Y_1^n|W_1,W_2,\alpha^n)}_{=~0} \nonumber \\
& ~+ \beta\left( p, \rho_{\sf Tx} \right) H(Y_2^n|\alpha^n) - \beta\left( p, \rho_{\sf Tx} \right) H(Y_2^n|W_2,\alpha^n) + n \epsilon_n \nonumber \\
& \overset{\textrm{Lemma}~\ref{Lemma:EntropyLeakage}}\leq \beta\left( p, \rho_{\sf Tx} \right) H(Y_2^n|\alpha^n) + n \epsilon_n \nonumber \\
& \leq n \beta\left( p, \rho_{\sf Tx} \right) (1-p^{\sf Rx}_{00}) + \epsilon_n.
\end{align}
where $\epsilon_n \rightarrow 0$ as $n \rightarrow \infty$, and the last step holds since ${\sf Rx}_2$ does not receive any signal $p^{\sf Rx}_{00}$ fraction of the time. Dividing both sides by $n$ and let $n \rightarrow \infty$, we get
\begin{align}
R_1 + \beta\left( p, \rho_{\sf Tx} \right) R_2 \leq \beta\left( p, \rho_{\sf Tx} \right) (1-p^{\sf Rx}_{00}).
\end{align}

\subsection{Packet Networks: An Extremal Rank-Ratio Inequality}
\label{Section:ConversePacket}

In this subsection we derive the outer-bounds on the throughput region of spatially correlated wireless packet networks. The following lemma is the equivalent of the Entropy Leakage Inequality, Lemma~\ref{Lemma:EntropyLeakage}, for wireless packet networks.

\begin{lemma}[Extremal Rank-Ratio Inequality for Spatially Correlated Channels]
\label{Lemma:Leakage}
For spatially correlated interference packet networks described in Section~\ref{Section:Problem} and for a non-zero $p \in \mathcal{S}_{\rho_{\sf Tx}} \cap \mathcal{S}_{\rho_{\sf Rx}}$, we have
\begin{align}
\mathbb{E}\left[ \mathrm{rank}\left[ \mathbf{G}_{21}^n \mathbf{V}_{1}^n \right] \right] \geq \frac{1}{\beta\left( p, \rho_{\sf Tx} \right)} \mathbb{E}\left[ \mathrm{rank}\left[ \mathbf{G}_{11}^n \mathbf{V}_{1}^n \right] \right].
\end{align}
where
\begin{align}
\beta\left( p, \rho_{\sf Tx} \right) = 1 + \left( 1 - \rho_{\sf Tx} \right) \left( 1 - p \right).
\end{align}
\end{lemma}

\begin{proof} 
\begin{align}
\mathbb{E} & \left[ \mathrm{rank}\left[ \mathbf{G}_{21}^n \mathbf{V}_{1}^n \right] \right] \nonumber \\
& \overset{(a)}= \mathbb{E}\left[ \sum_{t=1}^n{\mathrm{rank}\left[ \mathbf{G}_{21}^t \mathbf{V}_{1}^t \right] - \mathrm{rank}\left[ \mathbf{G}_{21}^{t-1} \mathbf{V}_{1}^{t-1} \right]} \right] \nonumber \\
& = \mathbb{E}\left[ \sum_{t=1}^n{\mathbf{1}{\left\{ \alpha_{21}(t) g_{21}(t) \vec{v}_{1}(t) \notin \mathrm{rowspan}\left( \mathbf{G}_{21}^{t-1} \mathbf{V}_{1}^{t-1} \right) \right\} }} \right] \nonumber \\
& \overset{(b)}= p\mathbb{E}\left[ \sum_{t=1}^n{\mathbf{1}{\left\{g_{21}(t) \vec{v}_{1}(t) \notin \mathrm{rowspan}\left( \mathbf{G}_{21}^{t-1} \mathbf{V}_{1}^{t-1} \right) \right\} }} \right] \nonumber \\
& \overset{(c)}= p \mathbb{E}\left[ \sum_{t=1}^n{\mathbf{1}{\left\{ \vec{v}_{1}(t)~\notin~\mathrm{rowspan}\left( \mathbf{G}_{21}^{t-1} \mathbf{V}_{1}^{t-1} \right) \right\} }} \right] \nonumber \\
& \overset{(d)}\geq p \mathbb{E}\left[ \sum_{t=1}^n{\mathbf{1}{\left\{ \vec{v}_{1}(t)~\notin~\mathrm{rowspan}\begin{bmatrix} \mathbf{G}_{21}^{t-1} \mathbf{V}_{1}^{t-1} \\ \mathbf{G}_{11}^{t-1} \mathbf{V}_{1}^{t-1} \end{bmatrix}  \right\} }} \right] \nonumber \\
& \overset{(e)}= \frac{p}{2p-pq\rho_{\sf Tx}-p^2} \times \nonumber \\
& ~~\mathbb{E}\left[ \sum_{t=1}^n{ \mathrm{rank} \begin{bmatrix} \mathbf{G}_{21}^t \mathbf{V}_{1}^t \\ \mathbf{G}_{11}^t \mathbf{V}_{1}^t \end{bmatrix} - \mathrm{rank} \begin{bmatrix} \mathbf{G}_{21}^{t-1} \mathbf{V}_{1}^{t-1} \\ \mathbf{G}_{11}^{t-1} \mathbf{V}_{1}^{t-1} \end{bmatrix} } \right] \nonumber 
\end{align}
\begin{align}
& \overset{(f)}= \frac{1}{1 + \left( 1 - \rho_{\sf Tx} \right) \left( 1 - p \right)} \mathbb{E}\left[ \mathrm{rank} \begin{bmatrix} \mathbf{G}_{21}^n \mathbf{V}_{1}^n \\ \mathbf{G}_{11}^n \mathbf{V}_{1}^n \end{bmatrix}\right] \nonumber \\
& \overset{(g)}\geq \frac{1}{\beta\left( p, \rho_{\sf Tx} \right)} \mathbb{E}\left[ \mathrm{rank} \begin{bmatrix} \mathbf{G}_{21}^n \mathbf{V}_{1}^n \\ \mathbf{G}_{11}^n \mathbf{V}_{1}^n \end{bmatrix}\right] \nonumber \\
& \geq \frac{1}{\beta\left( p, \rho_{\sf Tx} \right)} \mathbb{E}\left[ \mathrm{rank}\left[ \mathbf{G}_{11}^n \mathbf{V}_{1}^n \right] \right],
\end{align}
where $(a)$ follows the definition of the $\mathrm{rank}$ function; $(b)$ holds since $\alpha_{21}(t)$ follows a Bernoulli $\mathcal{B}\left( p \right)$ distribution; $(c)$ holds since if $g_{21}(t) = 0$, then $\alpha_{21}(t) = 0$; $(d)$ follows the fact that adding $\mathbf{G}_{11}^{t-1} \mathbf{V}_{1}^{t-1}$ enlarges the $\mathrm{rowspan}$; $(e)$ holds since the probability that at least one of the shadowing coefficients associated with ${\sf Tx}_1$ is equal to $1$ is given by
\begin{align}
1 - p^{\sf Tx}_{00} = 2p-pq\rho_{\sf Tx}-p^2;
\end{align}
$(f)$ holds since $p \neq 0$; $(g)$ follows the definition of $\beta\left( p, \rho_{\sf Tx} \right)$.
\end{proof}

To derive the outer-bounds, we again mimic the proof for the finite-field model. We have
\begin{align}
\label{Eq:OuterBound}
n & \left( R_1 + \beta\left( p, \rho_{\sf Tx} \right) R_2 \right) \nonumber \\
& \overset{a.s.}=  \mathbb{E}\left[ \mathrm{dim}\left( \mathrm{Proj}_{\mathcal{I}_1^c}~\mathrm{colspan}\left( \mathbf{G}_{11}^n \mathbf{V}_{1}^n \right) \right) \right] \nonumber \\
& ~~+ \beta\left( p, \rho_{\sf Tx} \right) \mathbb{E}\left[ \mathrm{dim}\left( \mathrm{Proj}_{\mathcal{I}_2^c}~\mathrm{colspan}\left( \mathbf{G}_{22}^n \mathbf{V}_{2}^n \right) \right) \right] \nonumber \\
& \overset{(a)}\leq \mathbb{E}\left[ \mathrm{rank}\left[ \mathbf{G}_{11}^n \mathbf{V}_{1}^n \right] \right] \nonumber \\
& ~~+ \beta\left( p, \rho_{\sf Tx} \right) \mathbb{E}\left[ \mathrm{rank}\left[ \mathbf{G}_{21}^n \mathbf{V}_{1}^n + \mathbf{G}_{22}^n \mathbf{V}_{2}^n \right] \right] \nonumber \\
& ~~- \beta\left( p, \rho_{\sf Tx} \right) \mathbb{E}\left[ \mathrm{rank}\left[ \mathbf{G}_{21}^n \mathbf{V}_{1}^n \right] \right] \nonumber \\
& \overset{\mathrm{Lemma~\ref{Lemma:Leakage}}}\leq \beta\left( p, \rho_{\sf Tx} \right) \mathbb{E}\left[ \mathrm{rank}\left[ \mathbf{G}_{21}^n \mathbf{V}_{1}^n + \mathbf{G}_{22}^n \mathbf{V}_{2}^n \right] \right] \nonumber \\
& \overset{(b)}\leq n \beta\left( p, \rho_{\sf Tx} \right) (1-p^{\sf Rx}_{00}), 
\end{align}
where the first equality is needed to guarantee (\ref{eq:decodability}) holds and each receiver can decode its intended packets; $(a)$ follows since we ignored interference at ${\sf Rx}_1$; $(b)$ holds since ${\sf Rx}_2$ does not receive any useful information $p^{\sf Rx}_{00}$ fraction of the time.

Dividing both sides of (\ref{Eq:OuterBound}) by $n$ and let $n \rightarrow \infty$, we get
\begin{align}
\label{Eq:SumRateBound}
R_1 + \beta\left( p, \rho_{\sf Tx} \right) R_2 \leq \beta\left( p, \rho_{\sf Tx} \right) (1-p^{\sf Rx}_{00}).
\end{align}
Similarly, we can obtain 
\begin{align}
\beta\left( p, \rho_{\sf Tx} \right) R_1 + R_2 \leq \beta\left( p, \rho_{\sf Tx} \right) (1-p^{\sf Rx}_{00}).
\end{align}

This completes the converse proof of Theorem~\ref{THM:CapacityCorrelated} and in the following section, we present our communication protocol.

\begin{remark}
The outer-bound obtained in (\ref{Eq:SumRateBound}) is not always active. For instance as discussed in Section~\ref{Section:Main} and depicted in Fig.~\ref{Fig:region-receiver}, for $p = 0.5$, $\rho_{\sf Tx} = 0$, and $\rho_{\sf Rx} \in \left[ -1, -1/3 \right]$, the throughput region is determined by the bounds on individual rates, \emph{i.e.} $0 \leq R_1, R_2 \leq 0.5$.
\end{remark}

\begin{figure}[ht]
\centering
\subfigure[]{\includegraphics[height = 1.3in]{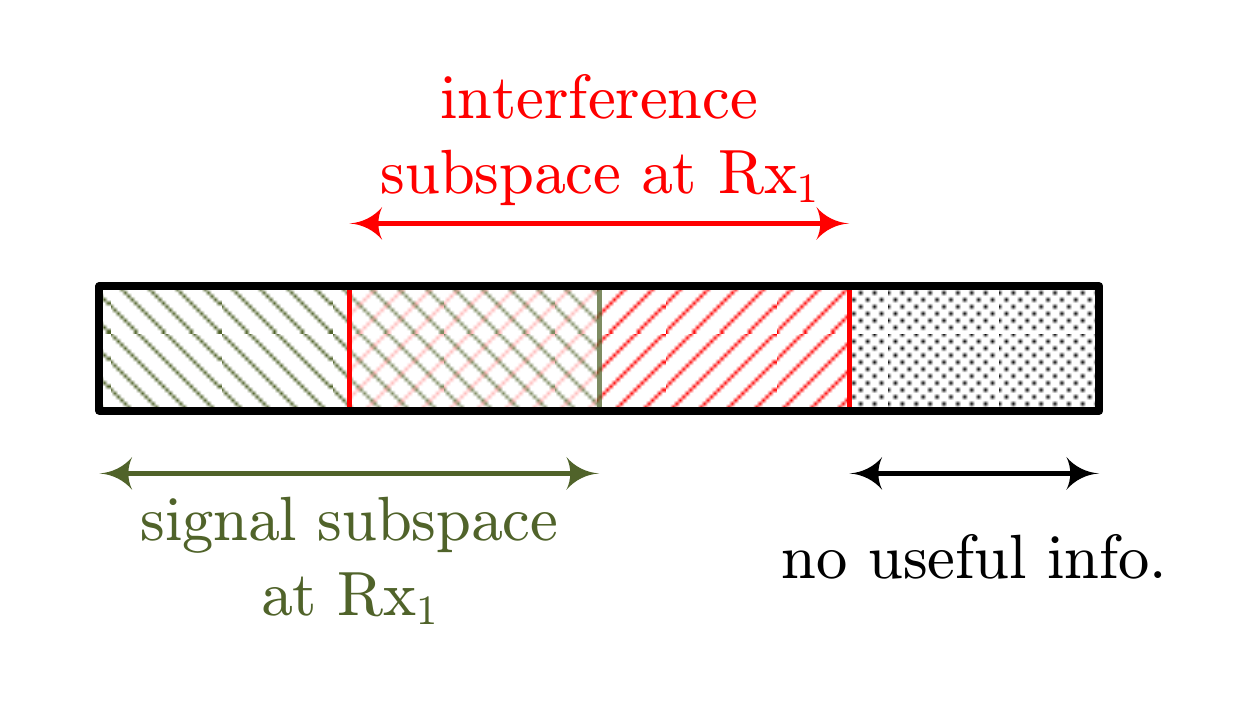}}
\subfigure[]{\includegraphics[height = 1.3in]{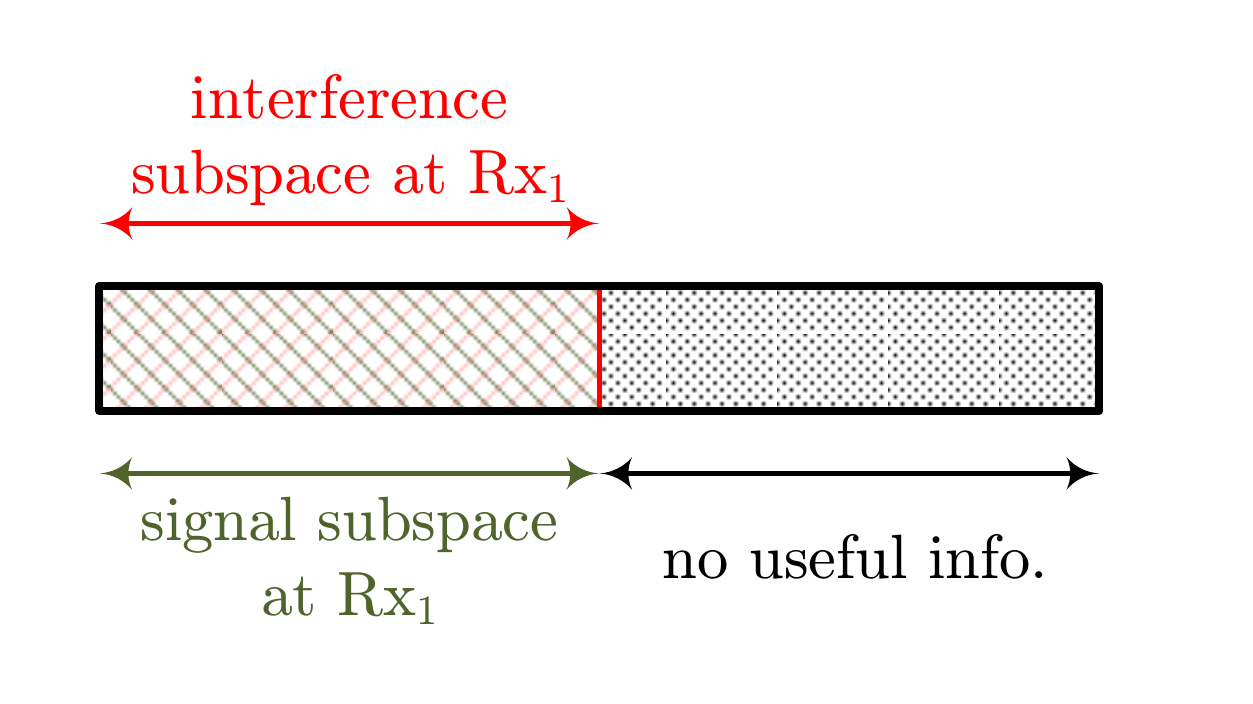}}
\subfigure[]{\includegraphics[height = 1.3in]{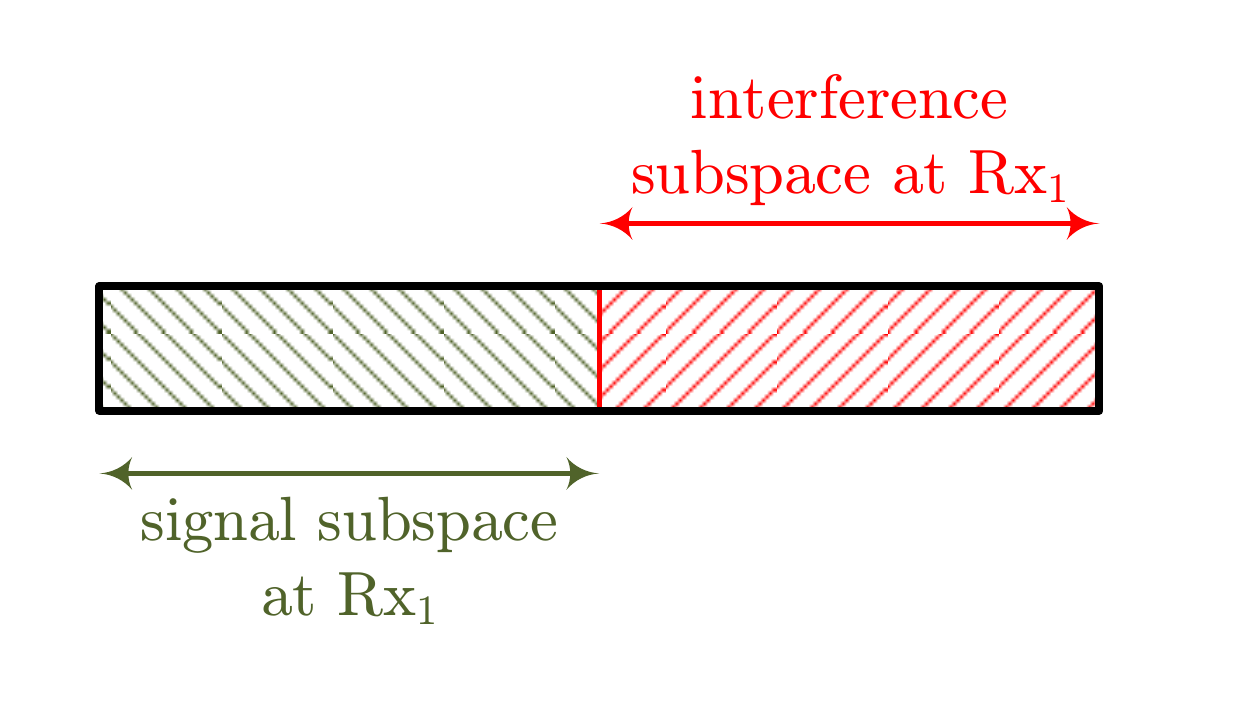}}
\caption{The available subspace at receiver ${\sf Rx}_1$ for $p = 0.5$ and: (a) $\rho_{\sf Rx} = 0$; (b) $\rho_{\sf Rx} = 1$; (c) $\rho_{\sf Rx} = -1$. A negative $\rho_{\sf Rx}$ pushes the signal subspace and the interference subspace away from each other.}\label{Fig:Subspace}
\end{figure}

The shape of the region is defined by $\beta\left( p, \rho_{\sf Tx} \right)$ which is a function of spatial correlation at the transmitters. Spatial correlation at the receivers plays an important role as well. A positive $\rho_{\sf Rx}$ pulls the signal subspace and the interference subspace closer to each other while a negative $\rho_{\sf Rx}$ pushes them away from each other. Fig.~\ref{Fig:Subspace} illustrates the available subspace at receiver ${\sf Rx}_1$ for $p = 0.5$ and $\rho_{\sf Rx} = -1, 0, 1$. For independent links at receivers $\rho_{\sf Rx} = 0$ as in Fig.~\ref{Fig:Subspace}(a), the signal subspace and the interference subspace overlap and $1/4$ of the available subspace to receiver ${\sf Rx}_1$ does not provide any useful information (this is when $\alpha_{11}(t) = \alpha_{21}(t) = 0$). When the links connected to ${\sf Rx}_1$ are fully correlated, \emph{i.e.} $\rho_{\sf Rx} = 1$, the signal subspace and the interference subspace completely overlap and  $1/2$ of the available subspace to receiver ${\sf Rx}_1$ does not provide any useful information as illustrated in Fig.~\ref{Fig:Subspace}(b). Finally when $\rho_{\sf Rx} = -1$ as in Fig.~\ref{Fig:Subspace}(c),  the signal subspace and the interference subspace are completely disjoint and the entire available subspace to receiver ${\sf Rx}_1$ is utilized. In this latter case since interference lies in a separate subspace from the signal, the task of interference management becomes simpler and the network throughput improves.


\section{Achievability Proof of Theorem~\ref{THM:CapacityCorrelated}: Communication Protocols for Spatially Correlated Channels}
\label{Section:Achievability}

In this section we describe an opportunistic communication protocol (or achievability strategy) that matches the outer-bounds, thus completing the proof of Theorem~\ref{THM:CapacityCorrelated}. We present this protocol for packet networks. The corresponding protocol for the finite-field model can be easily recovered by setting all continuous channel gains, $g_{ji}(t)$'s, equal to $1$. 


We first present two examples to clarify the key ideas of the communication protocol. These examples demonstrate the impact of (positive and negative) spatial correlation at both the receiver side and the transmitter side. As we mentioned in Remark~\ref{Remark:Family}, fixing $p, \rho_{\sf Tx},$ and $\rho_{\sf Rx}$ does not specify a unique channel but rather a family of channels. An interesting aspect of our communication protocol is that it only relies on the correlation coefficients rather than specific channel realizations and works for all channels within a given family. 

The throughput region for spatially independent channels ($\rho_{\sf Tx} = 0$ and $\rho_{\sf Rx} = 0$) is given in~\cite{AlirezaBFICDelayed}. Our communication protocol is the generalization of that result to spatially correlated channels. As we describe the protocol for the two examples, we also highlight the differences between our protocol and the communication protocol of~\cite{AlirezaBFICDelayed}. Once the required modifications are described and justified, same changes can be applied to for any choice of $p, \rho_{\sf Tx},$ and $\rho_{\sf Rx}$ as discussed at the end of this section.

\subsection{Example~1:~Communication Protocol for $p = \rho_{\sf Tx} = \rho_{\sf Rx} = 0.5$}
\label{Section:Example1}

In this example we assume that the wireless links both at the transmitters and at the receivers are positively correlated, and we present the achievability strategy for the maximum symmetric sum-rate point given by
\begin{align}
\label{Eq:MaxSumRate}
\left. \frac{2\beta\left( p, \rho_{\sf Tx} \right)\left( 1 - p^{\sf Rx}_{00} \right)}{1+\beta\left( p, \rho_{\sf Tx} \right)} \right|_{p=\rho_{\sf Tx}=\rho_{\sf Rx} =0.5}= \frac{25}{36}.
\end{align}

Suppose each transmitter wishes to communicate $m$ packets to its intended receiver. It suffices to show that this task can be accomplished (with vanishing error probability as $m \rightarrow \infty$) in
\begin{align}
\label{Eq:TotalTime}
& \left. \frac{1+\beta\left( p, \rho_{\sf Tx} \right)}{\beta\left( p, \rho_{\sf Tx} \right)\left( 1 - p^{\sf Rx}_{00} \right)}\right|_{p=\rho_{\sf Tx}=\rho_{\sf Rx} =0.5} m + \mathcal{O}\left( m^{\frac{2}{3}} \right) \nonumber \\
& ~= \frac{72}{25} m + \mathcal{O}\left( m^{\frac{2}{3}} \right)
\end{align}
time instants. The protocol is divided into two phases described below.

\begin{figure*}[ht]
\centering
\includegraphics[height = 3in]{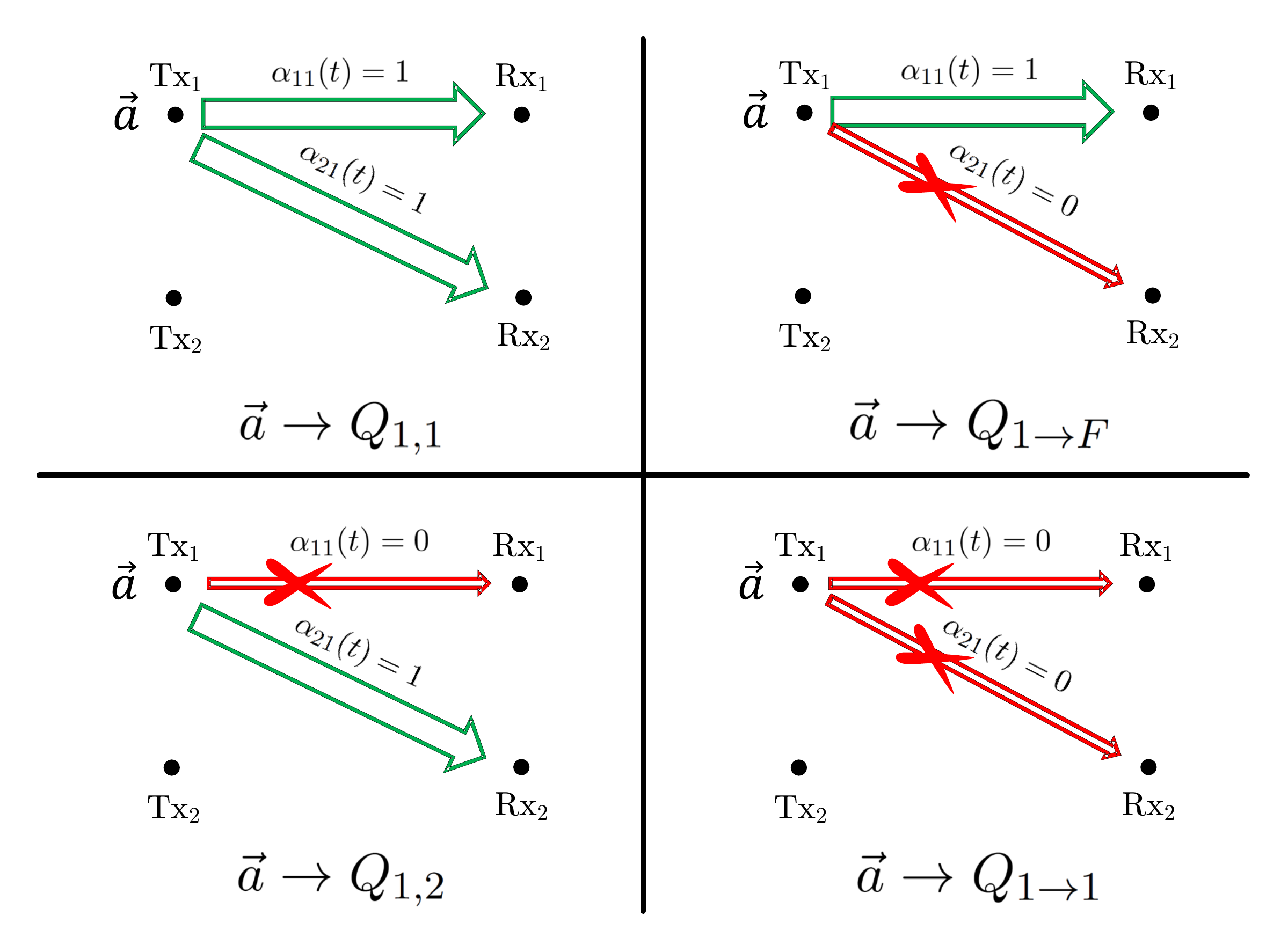}
\caption{Based on the shadowing coefficients at the time of communication, the status of packet $\vec{a}$ is updated. We note that the shadowing coefficients are learned with unit delay.\label{Fig:TableFig}}
\end{figure*}

\vspace{1.5mm}

\noindent {\bf Phase 1}: At the beginning of the communication block, we assume that the $m$ packets at ${\sf Tx}_i$ are in a queue\footnote{We assume that the queues are column vectors and packets are placed in each queue according to the order they join the queue.} denoted by $Q_{i \rightarrow i}$, $i=1,2$. At each time instant $t$, ${\sf Tx}_i$ transmits a packet from $Q_{i \rightarrow i}$ and this packet will either stay in this initial queue or will transition to one of the queues listed in Fig.~\ref{Fig:TableFig} for ${\sf Tx}_1$ and $\vec{a}$. If at time instant $t$, $Q_{i \rightarrow i}$ is empty, then ${\sf Tx}_i$, $i=1,2$, remains silent until the end of Phase 1.
\begin{enumerate}
\item [(A)] $Q_{i \rightarrow F}$: The packets for which no retransmission is required and thus we consider delivered. It is worth mentioning that any interference from ${\sf Tx}_{\bar{i}}$ at ${\sf Rx}_i$ will be resolved by that transmitter as we describe later;

\item [(B)] $Q_{i, 1}$: The packets for which at the time of communication both shadowing coefficients of the links connected to ${\sf Tx}_i$ were equal to $1$;

\item [(C)] $Q_{i, 2}$: The packets for which at the time of communication we have $\alpha_{ii}(t) = 0$ and $\alpha_{\bar{i}i}(t) = 1$.
\end{enumerate}

For $p = 0.5$ and $\rho_{\sf Tx} = \rho_{\sf Rx} = 0.5$, we have
\begin{align}
& p^{\sf Tx}_{11} = p^{\sf Rx}_{11} = \frac{3}{8}, \nonumber \\
& p^{\sf Tx}_{01} = p^{\sf Tx}_{10} = p^{\sf Rx}_{01} = p^{\sf Rx}_{10} \frac{1}{8}, \nonumber \\
& p^{\sf Tx}_{00} = p^{\sf Rx}_{00} = \frac{3}{8},
\end{align}
where the notation is defined in (\ref{Eq:RangeTx}) and (\ref{Eq:RangeRx}). Phase~$1$ continues for 
\begin{align}
\frac{1}{1-p^{\sf Tx}_{00}} m + m^{\frac{2}{3}} = \frac{8}{5} m + m^{\frac{2}{3}}
\end{align}
time instants and if at the end of this phase either of the queues $Q_{i \rightarrow i}$ is not empty, we declare error type-I and halt the transmission (for simplicity, we assume $m$ is chosen such that $m^{\frac{2}{3}} \in \mathbb{Z}$). 

\vspace{1.5mm}

\noindent {\bf Comparison to the communication protocol for independent links:} Here, spatial correlation at the transmitter side changes the length of Phase~1. For $\rho_{\sf Tx} = 0$, we have $p^{\sf Tx}_{00} = (1-p)^2 = 1/4$.

Assuming that the transmission is not halted, let $N_{i,1}$ and $N_{i, 2}$ denote the number of packets in queues $Q_{i,1}$ and $Q_{i,2}$ respectively at the end of the first phase, $i=1,2$. The transmission strategy will be halted and error type-II occurs, if any of the following events happens.
\begin{align}
\label{eq:errortypeII}
& N_{i,1} > \mathbb{E}[N_{i,1}] + 2 m^{\frac{2}{3}} \overset{\triangle}= n_{i,1}, \quad i=1,2; \nonumber \\
& N_{i,2} > \mathbb{E}[N_{i,2}] + 2 m^{\frac{2}{3}} \overset{\triangle}= n_{i,2}, \quad i=1,2.
\end{align}

From basic probability, we have
\begin{align}
\label{eq:expectedvalues}
\mathbb{E}[N_{i,1}] = \frac{p^{\sf Tx}_{11}m}{1 - p^{\sf Tx}_{00}} = \frac{3}{5}m,~\mathbb{E}[N_{i,2}] = \frac{p^{\sf Tx}_{01}m}{1 - p^{\sf Tx}_{00}} = \frac{1}{5}m.
\end{align}
At the end of Phase $1$, we add deterministic packets (if necessary) in order to make queues $Q_{i,1}$ and $Q_{i,2}$ of size equal to $n_{i,1}$ and $n_{i,2}$ respectively as given above, $i=1,2$. 

Statistically a fraction
\begin{align}
\frac{p^{\sf Rx}_{01}}{p} = \frac{1}{4}
\end{align}
of the packets in $Q_{i,1}$ and the same fraction of the bits in $Q_{i,2}$ are known to ${\sf Rx}_{\bar{i}}$, $i=1,2$. Denote the number of bits in $Q_{i,j}$ known to ${\sf Rx}_{\bar{i}}$ by
\begin{align}
N_{i,j|{\sf Rx}_{\bar{i}}}, \qquad i,j \in \{ 1, 2 \}.
\end{align}
At the end of communication, if we have
\begin{align}
\label{Eq:KnownBits}
N_{i,j|{\sf Rx}_{\bar{i}}} < \frac{p^{\sf Rx}_{01}}{p} n_{i,j} - 2 m^{\frac{2}{3}} = \frac{1}{4} n_{i,j} - 2 m^{\frac{2}{3}}, \qquad i,j \in \{ 1, 2 \},
\end{align}
we declare error type-III. 

Using Bernstein inequality~\cite{bernstein1924modification}, we can show that the probability of errors of types I, II, and III decreases exponentially and approaches zero as $m \rightarrow \infty$\footnote{In simple terms this inequality states that if $X_1,\ldots,X_r$ are $r$ independent random variables, and $M=\sum_{i=1}^r{X_i}$, then $Pr\left[ |M-\mathbb{E}\left[ M \right] | > \alpha \right] \leq 2 \exp \left( \frac{-\alpha^2}{4 \sum_{i=1}^r \mathrm{Var} \left( X_i \right)} \right)$.}. In fact, throughout this section, we intentionally picked $m^{\frac{2}{3}}$ to add to the constants in order to guarantee that the error terms vanish as $m$ increases. For instance, to bound the probability of error type-I, we have
\begin{align}
& \Pr \left[ \mathrm{error~type \noindent - \noindent I} \right] \overset{\text{Union~Bound}}\leq  \sum_{i=1}^2{\Pr \left[ Q_{i \rightarrow i} \mathrm{~is~not~empty} \right]} \nonumber \\
&~= 4 \exp \left( \frac{-m^{4/3}}{4 \left( 1 - \left( p^{\sf Tx}_{00} \right)^2 \right) \left( p^{\sf Tx}_{00} \right)^2 \left[ \frac{1}{1-p^{\sf Tx}_{00}} m + m^{\frac{2}{3}} \right]} \right),
\end{align}
which decreases exponentially to zero as $m \rightarrow \infty$. We refer the reader to~\cite{AlirezaBFICDelayed} for a more detailed discussion on using Bernstein inequality (and Chernoff-Hoeffding bound~\cite{Chernoff,Hoeffding}) to analyze the error probabilities. For the rest of this subsection, we assume that Phase~1 is completed and no error has occurred.

\vspace{1.5mm}

\noindent {\bf Intuition}: At this point it is worth describing the logic behind the following phase. The value of wireless is simultaneous communication to many nodes. Thus, the goal is to create linear combinations of packets that are of interest to \emph{both} receivers. 

\vspace{1.5mm}

\noindent {\bf Phase 2}: During this phase we deliver sufficient number of linearly independent combinations of packets to each receiver so that they can decode their corresponding packets. A linear combination of packets is the bitwise modulo summation of the corresponding codewords in the binary field. We create these linear combinations in such a way that they are of interest to both receivers. Denote by $Q^c_{1,1}$ and $Q^c_{1,2}$ the fraction of the packets in $Q_{1,1}$ and $Q_{1,2}$ respectively for which at the time of transmission $\alpha_{22}(t) = 1$, and by $Q^{nc}_{1,1}$ and $Q^{nc}_{1,2}$ the fraction of the packets in $Q_{1,1}$ and $Q_{1,2}$ respectively for which at the time of transmission $\alpha_{22}(t) = 0$. Similar definitions apply to $Q^c_{2,1}, Q^c_{2,2}, Q^{nc}_{2,1}$ and $Q^{nc}_{2,2}$. We have
\begin{align}
\label{eq:expectedvalues2}
\mathbb{E}[N^c_{i,1}] =  \frac{9}{20}m,~\mathbb{E}[N^c_{i,2}] =  \frac{3}{20}m,\quad i=1,2.
\end{align}
  
\begin{figure}[ht]
\centering
\subfigure[]{\includegraphics[height = 1.5in]{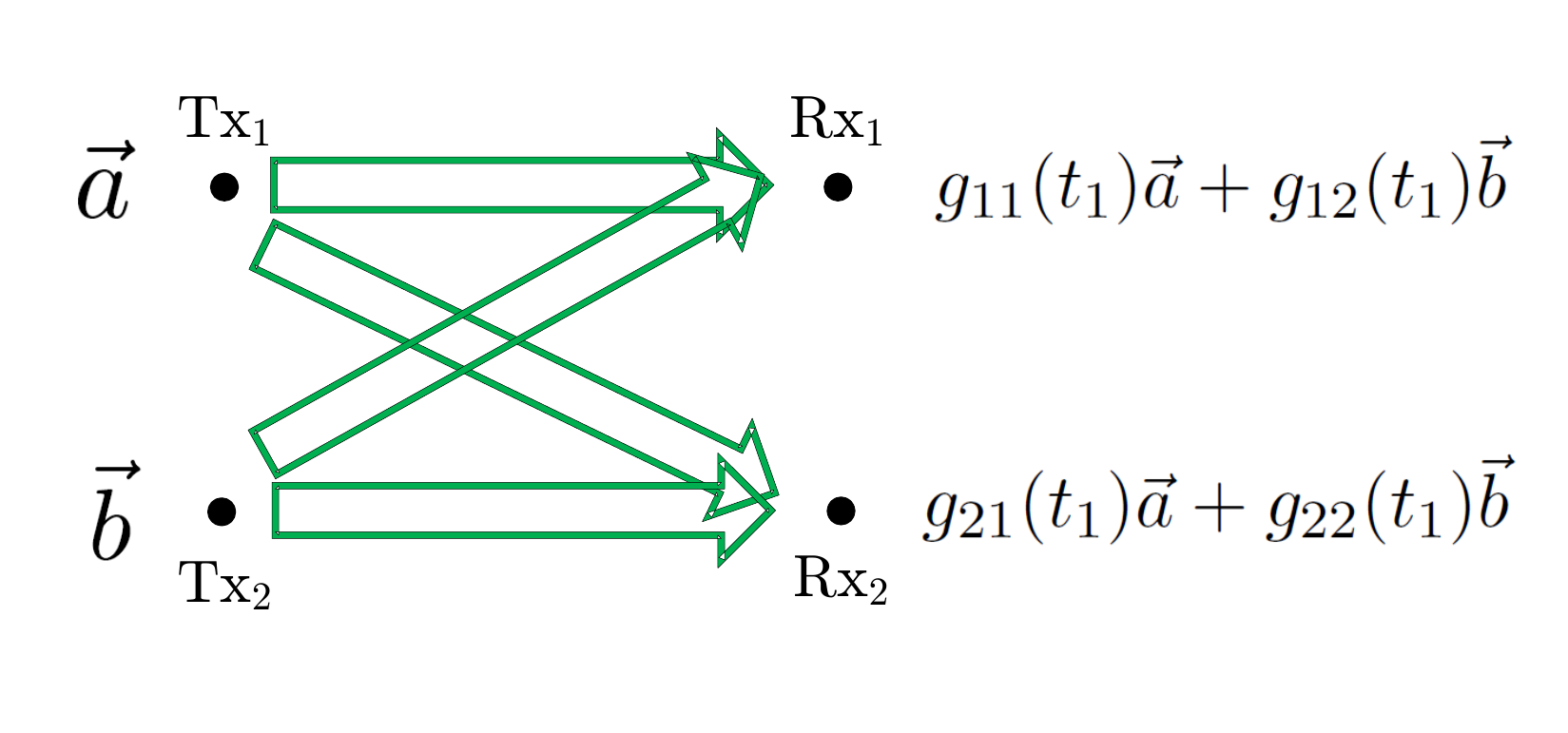}}
\hspace{.5in}
\subfigure[]{\includegraphics[height = 1.5in]{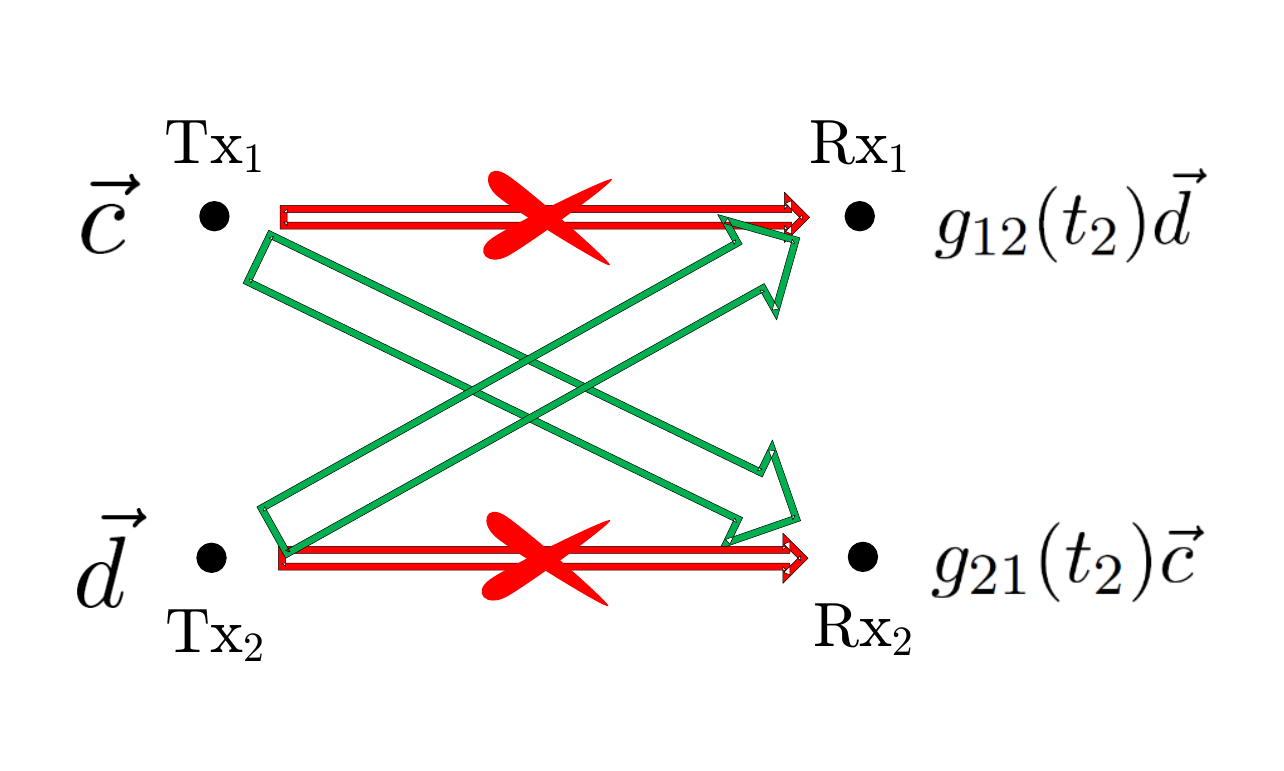}}
\caption{Packets in $Q^{nc}_{i,2}$ can be combined with packets in $Q^{c}_{i,1}$ to create packets of common interest: $\vec{a} + \vec{c}$ and $\vec{b} + \vec{d}$ are of common interest to both receivers.}\label{Fig:Combine1}
\end{figure}	
		
Packets in $Q^{nc}_{i,2}$ can be combined with packets in $Q^{c}_{i,1}$ to create packets of common interest as depicted in Fig.~\ref{Fig:Combine1}. Packets in $Q^{c}_{i,2}$ are needed at both receivers as depicted in Fig.~\ref{Fig:Combine2} (no combination in this case). Finally, the packets in $Q^{c}_{i,1}$ minus the ones combined with packets in $Q^{nc}_{i,2}$ are needed at both receivers, \emph{e.g.}, packet $\vec{a}$ in Fig.~\ref{Fig:Combine1}(a) is useful for both receivers. However for these remaining packets in $Q^{c}_{i,1}$ and $Q^{c}_{\bar{i},1}$ only half of them need to be provided to both receivers. As a result, the expected number of total linearly independent combinations needed at both receivers is given by:
\begin{align}
2 \times \underbrace{\frac{3}{20}m}_{Q^{c}_{i,2}} + 2 \times \underbrace{\frac{1}{20}m}_{Q^{nc}_{i,2}~\text{combined~w/~}Q^{c}_{i,1}} + \underbrace{\frac{8}{20}m}_{\text{remaining~in~}Q^{c}_{i,1}} = \frac{4}{5} m.
\end{align}

\begin{figure}[ht]
\centering
\includegraphics[height = 1.5in]{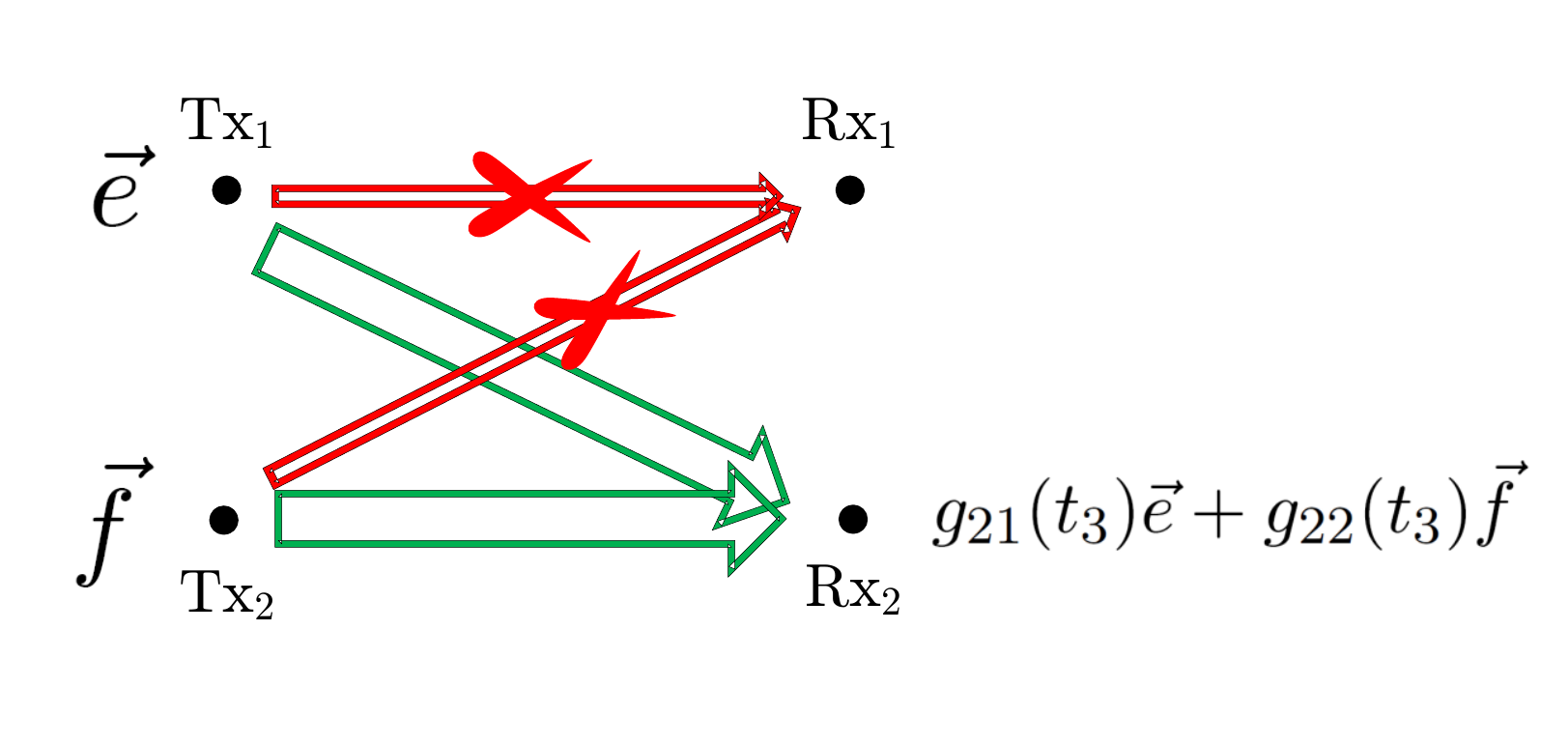}
\caption{Packets in $Q^{c}_{i,2}$ are needed at both receivers. For example, packet $\vec{e}$ is needed at ${\sf Rx}_1$ and helps ${\sf Rx}_2$ to remove interference and decode its intended packet $\vec{f}$.}\label{Fig:Combine2}
\end{figure}

\noindent {\bf Comparison to the communication protocol for independent links:} If we keep $\rho_{\sf Tx} = 0.5$ but change $\rho_{\sf Rx}$ to $0$ (\emph{i.e.} independent links at the receivers), the expected number of equations needed at both receivers reduces to $3m/5$. The reason is that with $\rho_{\sf Rx} = 0.5$ more packets are in $Q^c_{1,1}$ and $Q^c_{1,2}$ which limits the opportunity to combine packets. In other words, positive spatial correlation at receivers pushes the signal subspace and the interference subspace closer to each other which results in requiring more information to resolve interference. 

We conclude that ${\sf Tx}_i$ only needs to create
\begin{align}
\frac{2}{5} m + \mathcal{O}\left(m^{\frac{2}{3}}\right)
\end{align}
linearly independent equations of the packets in $Q_{i,1}$ and $Q_{i,2}$, $i=1,2$, and deliver them to \emph{both} receivers. This problem of delivering packets to both receivers resembles a network with two transmitters and two receivers where each transmitter ${\sf Tx}_i$ wishes to communicate an independent message $\hbox{W}_i$ to {\it both} receivers as depicted in Fig.~\ref{fig:two-multicast}, $i=1,2$. The channel gain model is the same as described in Section~\ref{Section:Problem}. We refer to this problem as the two-multicast problem. It is a straightforward exercise to show that for this problem, a rate-tuple of 
\begin{align}
\label{Eq:MulticastRate}
\left( R_1, R_2 \right) = \left( \frac{\left( 1- p^{\sf Rx}_{00} \right)}{2}, \frac{\left( 1- p^{\sf Rx}_{00} \right)}{2} \right) = \left( \frac{5}{16}, \frac{5}{16} \right)
\end{align} 
is achievable, see~\cite{AlirezaBFICDelayed}.

\vspace{1.5mm}

\noindent {\bf Comparison to the communication protocol for independent links:} For $\rho_{\sf Rx} = 0$ (\emph{i.e.} independent links at the receivers), the two-multicast problem with $p=0.5$ has a maximum achievable sum-rate of $3/4$ compared to $5/8$ for $\rho_{\sf Rx} = 0.5$. 

\begin{figure}[ht]
\centering
\includegraphics[height = 3.5 cm]{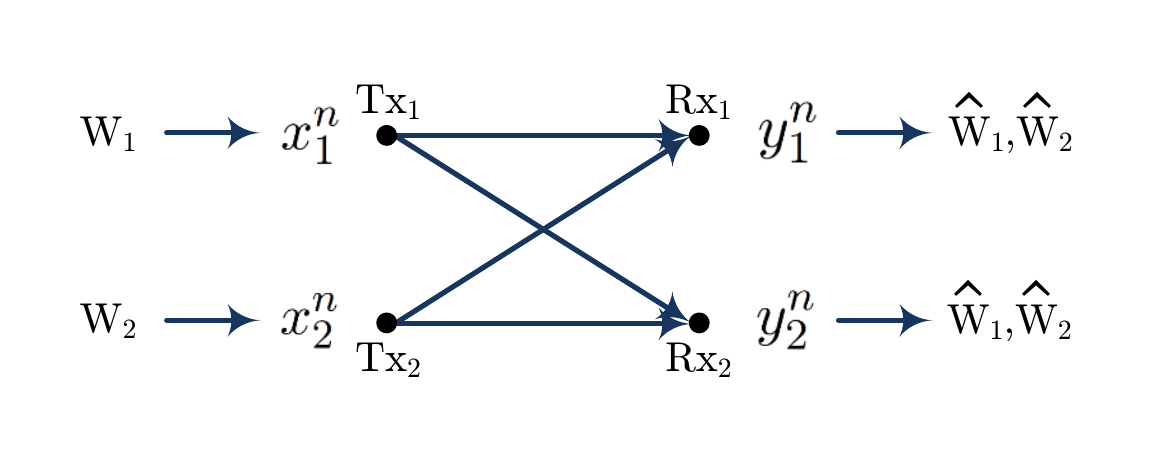}
\caption{Two-multicast network: ${\sf Tx}_i$ wishes to reliably communicate message $\hbox{W}_i$ to both receivers, $i=1,2$.\label{fig:two-multicast}}
\end{figure}

Using the communication protocol of the two-multicast problem, Phase~2 lasts for 
\begin{align}
\underbrace{\frac{8}{5}}_{\text{one~over~sum-rate~of~\ref{Eq:MulticastRate}}} \times \underbrace{\frac{4}{5}m}_{\text{\#~of~eqs~needed}} + \mathcal{O}\left(m^{\frac{2}{3}}\right) = \frac{32}{25} m + \mathcal{O}\left(m^{\frac{2}{3}}\right)
\end{align} 
time instants. If any error takes place, we declare error and terminate the communication. Thus to continue, we assume that the transmission is successful and no error has occurred.  

\vspace{1.5mm}

\noindent {\bf Decoding}: The idea is that upon completion of the communication protocol each receiver has sufficient number of linearly independent equations of the desired packets to decode them. As we described in detail, we designed the scheme in such way to guarantee this happens. 

\vspace{1.5mm}

\noindent {\bf Achievable rates}: Assuming that no error occurs, the total communication time is then equal to the length of Phases~1 and~2. Thus asymptotically, the total communication time is:
\begin{align}
\underbrace{\frac{8}{5} m}_{\text{Phase~1}} + \underbrace{\frac{32}{25} m}_{\text{Phase~2}} + \mathcal{O}\left(m^{\frac{2}{3}}\right) = \frac{72}{25} m + \mathcal{O}\left(m^{\frac{2}{3}}\right)
\end{align}
time instants which is what we expected as given in~(\ref{Eq:TotalTime}). This completes the communication protocol for $p = 0.5$ and $\rho_{\sf Tx} = \rho_{\sf Rx} = 0.5$.

\begin{remark}[Binary vs. Larger Finite Fields]
\label{Remark:BinaryField}
In this work we focus on the binary coefficients and the reason is as follows. As we saw in this example, for the transmission strategy what matters is whether or not transmit signal $X_i(t)$ contributes to the receiver observation $Y_j(t)$. For instance, in Fig.~\ref{Fig:Combine2} the key is that ${\sf Rx}_2$ has a linear combination of $\vec{e}$ and $\vec{f}$, and given this information, the continuous channel gains, $g_{21}(t_3), g_{22}(t_3)$, become irrelevant. In other words, even if we consider a larger finite field for the shadowing coefficients, all that matters is whether each channel gain is zero or not which is essentially a binary random variable. 
\end{remark}

\subsection{Example~2:~Communication Protocol for $p = 0.45,$ $\rho_{\sf Tx} = 0$ and $\rho_{\sf Rx} = -0.75$}

This example focuses on the impact of negative spatial correlation at the receiver side and assumes independent links at the transmitters. For this particular choice of channel parameters, we can achieve
\begin{equation}
\left\{ \begin{array}{ll}
\vspace{1mm} 0 \leq R_1 \leq 0.45, &\\
0 \leq R_2 \leq 0.45. &
\end{array} \right\}
\end{equation}
which is maximum possible throughput regardless of the assumption on the available channel state information (since each individual link has a maximum throughput of $p = 0.45$).

The achievability strategy is divided into three phases  described below compared to the two phases of the previous example.

\vspace{1.5mm}

\noindent {\bf Phase 1}: This phase follows the same steps as the one introduced in the previous subsection. In this example the wireless links connected to each transmitter are uncorrelated and as a result, Phase~$1$ continues for 
\begin{align}
\frac{1}{1-q^2} m + m^{\frac{2}{3}} = \frac{1}{1-.55^2} m + m^{\frac{2}{3}}  \left(\approx 1.4337 m + m^{\frac{2}{3}}\right)
\end{align}
time instants\footnote{Numbers are rounded in this example.}, and if at the end of this phase either of the queues $Q_{i \rightarrow i}$ is not empty, we declare error type-I and halt the transmission. 

\vspace{1.5mm}

\noindent {\bf Phase 2}: Similar to the previous example, the goal of this phase is to combine packets and use multicasting to deliver them. However, the large negative correlation at the receiver side results in fewer opportunities to combine packets. Denote by $Q^c_{1,1}$ and $Q^c_{1,2}$ the fraction of the packets in $Q_{1,1}$ and $Q_{1,2}$ respectively for which at the time of transmission $\alpha_{22}(t) = 1$, and by $Q^{nc}_{1,1}$ and $Q^{nc}_{1,2}$ the fraction of the packets in $Q_{1,1}$ and $Q_{1,2}$ respectively for which at the time of transmission $\alpha_{22}(t) = 0$. In this example $Q^c_{1,1}$ only contains a small fraction of the packets in $Q_{1,1}$. More precisely, only a fraction $p_{11}^{\sf Rx}/\left( p_{10}^{\sf Rx} + p_{11}^{\sf Rx}\right) \approx 0.0376$ of the packets in $Q_{1,1}$ are placed in $Q^c_{1,1}$. We combine the packets and retransmit them similar to the previous example, and a large fraction of packets will wait for Phase~3.

\vspace{1.5mm}

\noindent {\bf Phase 3}: In this phase the packets left in $Q^{nc}_{1,1}$ and $Q^{nc}_{1,2}$ are retransmitted. Since these packets are already available at the unintended receivers as shown in Fig.~\ref{Fig:Combine1}(b), they no longer create any interference when retransmitted and thus they can be delivered at a rate of $p = 0.45$ from each transmitter to its intended receiver. 

\vspace{1.5mm}

\noindent {\bf Achievable rates}: We skip decoding for this example and we just calculate the overall throughput. The total communication time is equal to:
\begin{align}
& \underbrace{\frac{1}{1-q^2} m}_{\text{Phase~1}} + \underbrace{\frac{p p_{11}^{\sf Rx}}{p_{10}^{\sf Rx} + p_{11}^{\sf Rx}} \left( 1 + p + \frac{p p_{11}^{\sf Rx}}{p_{10}^{\sf Rx} + p_{11}^{\sf Rx}} \right) m}_{\text{Phase~2}} \nonumber \\
& ~+  \underbrace{\frac{qm}{1-q^2} - \left( 1 - \frac{p_{11}^{\sf Rx}}{p_{10}^{\sf Rx} + p_{11}^{\sf Rx}} \right) m}_{\text{Phase~3}} + \mathcal{O}\left(m^{\frac{2}{3}}\right) \nonumber \\
& =  \frac{20}{9} m + \mathcal{O}\left(m^{\frac{2}{3}}\right),
\end{align}
which immediately implies that an overall throughput of $2p = 0.9$ is achievable in this example ($p = 0.45,$ $\rho_{\sf Tx} = 0$ and $\rho_{\sf Rx} = -0.75$).

\vspace{1.5mm}

\noindent {\bf Comparison to the communication protocol for independent links:} As mentioned before, a negative $\rho_{\sf Rx}$ pushes the signal subspace and the interference subspace away from each other. In this example we have $\rho_{\sf Rx} = -0.75$ and as a result, most of the packets that arrive at the unintended receiver do not cause any interference. Since these packets are known at the unintended receiver, retransmission can be accomplished at the maximum point-to-point rate of $p = 0.45$. In this example we could not create many packets of common interest which sounds discouraging. However, since the signal and the interference subspaces are nicely separated, we could achieve the maximum possible rates. 

\subsection{Communication Protocol for General $p,\rho_{\sf Tx}$ and $\rho_{\sf Rx}$}


The communication protocol for spatially independent links, \emph{i.e.} $\rho_{\sf Tx} = 0$ and $\rho_{\sf Rx} = 0$, is presented in~\cite{AlirezaBFICDelayed} and in essence, our communication protocol builds upon that result with careful modifications to incorporate spatial correlation. Intuitively, spatial correlation at the transmitters and at the receivers affect the communication protocol in the following ways:
\begin{enumerate}
\item Spatial correlation at the transmitters ($\rho_{\sf Tx}$) affects the length of Phase~1: positive correlation increases the length of Phase~1 whereas negative correlation decreases it. Moreover, $\rho_{\sf Tx}$ changes the distribution of packets in the queues of Fig.~\ref{Fig:TableFig}: positive correlation increases the number of packets in $Q_{1,1}$ whereas negative correlation increases the number of packets in $Q_{1,2}$. The size of these queues define the mixing procedure of the following communication phases;  

\item Spatial correlation at the receivers ($\rho_{\sf Rx}$) affects the number of interfering packets:  positive correlation increases the number of packets in $Q^c_{1,1}$ and $Q^c_{1,2}$ whereas negative correlation decreases this number. In other words, a positive $\rho_{\sf Rx}$ pushes the signal subspace and the interference subspace closer to each other while a negative $\rho_{\sf Rx}$ pushes them away from each other. Moreover, $\rho_{\sf Rx}$ directly affects the throughput of the two-multicast problem of Fig.~\ref{fig:two-multicast}: positive correlation decreases the throughput whereas negative correlation increases the throughput.
\end{enumerate}

Now that we understand intuitively how spatial correlation affects the communication protocol, consider a spatially correlated network with $p, \rho_{\sf Tx},$ and $\rho_{\sf Rx}$. In general, at each time the network can be in any of the $16$ possible realizations where state $k$ may occur with probability $0 \leq p_{{\sf S}_k} \leq 1$, $\sum_{k=1}^{16}{p_{{\sf S}_k}} = 1$, $k=1,2,\ldots,16$. In the two examples we presented above, the communication protocol at each transmitter relies on the delayed knowledge of the links connected to it and the statistics of the links connected to the other transmitter, see Fig.~\ref{Fig:TableFig}. But since at each time transmitters know all the past channel realizations, rather than relying on statistics we can use this knowledge to form the queues as described below. 
\begin{itemize}
\item $Q^c_{1,1}$: A packet from transmitter ${\sf Tx}_1$ moves to queue  $Q^c_{1, 1}$ if at the time of transmission the channel realization is ${\sf S}_k$ for $k \in \{ 1,2 \}$. For ${\sf Tx}_2$ and queue  $Q^c_{2, 1}$ this transition happens when $k \in \{ 1,3 \}$;

\item $Q^{nc}_{1,1}$: A packet from transmitter ${\sf Tx}_1$ moves to queue  $Q^{nc}_{1, 1}$ if at the time of transmission the channel realization is ${\sf S}_k$ for $k \in \{ 6,8 \}$. For ${\sf Tx}_2$ and queue  $Q^{nc}_{2, 1}$ this transition happens when $k \in \{ 10, 12 \}$;

\item $Q^c_{i,2}$: The status of a packet from transmitter ${\sf Tx}_1$ is updated to $Q^c_{i,2}$ if at the time of transmission the channel realization is ${\sf S}_k$ for $k \in \{ 11,12  \}$. For ${\sf Tx}_2$ and $Q^c_{2,2}$ the transition happens if $k \in \{ 7,8 \}$;

\item $Q^{nc}_{i,2}$: The status of a packet from transmitter ${\sf Tx}_1$ is updated to $Q^c_{i,2}$ if at the time of transmission the channel realization is ${\sf S}_k$ for $k \in \{ 14,15  \}$. For ${\sf Tx}_2$ and $Q^c_{2,2}$ the transition happens if $k \in \{ 13,15 \}$.


\end{itemize}
In the two examples we provided in this section, we relied on the statistics of the wireless links to form $Q^c_{i,1}, Q^{nc}_{i,1}, Q^c_{i,2},$ and $Q^{nc}_{i,2}$ but since global CSIT is available, we can actually perform this task by looking at the specific channel realizations as we did above. Once we have formed these queues, we combine them according to the following principles.

Packets in $Q^{nc}_{i,2}$ can be combined with packets in $Q^{c}_{i,1}$ to create packets of common interest as depicted in Fig.~\ref{Fig:Combine1}. Packets in $Q^{c}_{i,2}$ are needed at both receivers as depicted in Fig.~\ref{Fig:Combine2} (no combination in this case). If after combining the packets $Q^{c}_{i,1}$ is not empty, such packets are useful for both receivers, \emph{e.g.}, packet $\vec{a}$ in Fig.~\ref{Fig:Combine1}(a). Moreover, for the remaining packets in $Q^{c}_{i,1}$ and $Q^{c}_{\bar{i},1}$ only half of them need to be provided to both receivers. We communicated the combined packets and half the remaining packets in $Q^{c}_{i,1}$ using the transmission strategy of the two-multicast problem depicted in Fig.~\ref{fig:two-multicast}. On the other hand, suppose after combining the packets queues $Q^{nc}_{1,1}$ and $Q^{nc}_{1,2}$ are not empty. Such packets are already available at the unintended receivers as shown in Fig.~\ref{Fig:Combine1}(b) and they no longer create any interference when retransmitted and thus they can be delivered at a rate of $p$ from each transmitter to its intended receiver. The other corner points of the capacity region can be achieved applying the same modification principles to the protocol presented in~\cite{AlirezaBFICDelayed}.


\section{Connection to Other Problems}
\label{Section:Discussion}

In this section we describe the connections between our results and two other well-known problems.

\subsection{Wireless Networks with Alternating Topology}

Wireless networks with alternating topology were studied in~\cite{AlirezaBFICDelayed} where each link could be on or off with some probability and capacity results were obtained for the case of instantaneous and delayed channel state information at the transmitters. Other results have considered the Degrees of Freedom of networks with alternating topology as well~\cite{sun2013topological}.

The shadowing coefficients $\alpha_{ij}(t)$'s in our work are designed to capture the relative strength of different signals. However, these coefficients can be interpreted as connectivity coefficients and thus, our work is closely related to the results on networks with alternating topology. Moreover since we limit the precoding to be linear, our results are related to the study of the linear DoF of wireless networks~\cite{lashgari2014linear,kao2017linear,lashgari2014blind,yang2017linear,mukherjee2017secure}.

\subsection{Wireless Packet Networks}

In the context of wireless packet networks it is known that when collision occurs, the receiver can still store its analog received signal and utilize it for decoding the packets in the future. This can be done via a variety of techniques studied in multi-user joint detection and interference cancellation for cellular networks (\emph{e.g.},~\cite{boutros2002iterative,verdu1998multiuser,reynolds2002blind,el2000iterative,poor2004iterative}). ZigZag decoding~\cite{gollakota2008zigzag} also demonstrates that interference decoding and successive interference cancellation can be accomplished in 802.11 MAC. ZigZag decoding is carried on from the perspective of one receiver and takes advantage of asynchronous communication. In contrast, our work considers a network with multiple receivers and the goal is to maximize the overall network throughput. This is solely done through the knowledge of past interference pattern and no data communication between the transmitters happen. Finally, we assumed that packets are perfectly synchronized and this takes away the gain of ZigZag decoding from our problem.


\section{Conclusion}
\label{Section:Conclusion}

We studied the throughput region of spatially correlated interference packet networks. We quantified the impact of spatial correlation at the transmitters by deriving an extremal rank-ratio inequality. To achieve the outer-bounds, we developed a new transmission protocol that could have as many as three phases of communication. Our communication protocol takes full advantage of delayed interference structure knowledge and the available spatial correlation side information.  We learned that spatial correlation at the transmitters defines the shape of the throughput region while spatial correlation at the receivers defines the size of this region. 

In this work we assumed that each transmitter starts with a large number of packets to communicate. An interesting direction is to study and to characterize the throughput region for stochastic packet arrival with potential delivery deadlines. The stability region of the problem is an interesting direction and authors in~\cite{pan2016stability} did so for the results of~\cite{AlirezaInfocom2014}.

%
%
%

\bibliographystyle{ieeetr}
\bibliography{bib_ARQSpatial}

\begin{thebibliography}{10}

\bibitem{vahid2016does}
A.~Vahid and R.~Calderbank, ``When does spatial correlation add value to
  delayed channel state information?,'' in {\em IEEE International Symposium on
  Information Theory (ISIT)}, pp.~2624--2628, IEEE, 2016.

\bibitem{vahidARQ}
A.~Vahid and R.~Calderbank, ``{ARQ} for interference packet networks,'' {\em to
  appear in the IEEE International Symposium on Information Theory (ISIT)},
  2018.

\bibitem{abramson1970aloha}
N.~Abramson, ``The {ALOHA} system: another alternative for computer
  communications,'' in {\em Proceedings of the November 17-19, 1970, fall joint
  computer conference}, pp.~281--285, ACM, 1970.

\bibitem{AlirezaInfocom2014}
A.~Vahid, M.~A. Maddah-Ali, and A.~S. Avestimehr, ``Communication through
  collisions: Opportunistic utilization of past receptions,'' in {\em INFOCOM},
  pp.~2553--2561, IEEE, 2014.

\bibitem{AlirezaBFICDelayed}
A.~Vahid, M.~A. Maddah-Ali, and A.~S. Avestimehr, ``Capacity results for binary
  fading interference channels with delayed {CSIT},'' {\em IEEE Transactions on
  Information Theory}, vol.~60, no.~10, pp.~6093--6130, 2014.

\bibitem{vahid2016two}
A.~Vahid and R.~Calderbank, ``Two-user erasure interference channels with local
  delayed {CSIT},'' {\em IEEE Transactions on Information Theory}, vol.~62,
  no.~9, pp.~4910--4923, 2016.

\bibitem{jolfaei1993new}
K.~Jolfaei, S.~Martin, and J.~Mattfeldt, ``A new efficient selective repeat
  protocol for point-to-multipoint communication,'' in {\em IEEE International
  Conference on Communications (ICC'93)}, vol.~2, pp.~1113--1117, IEEE, 1993.

\bibitem{dana2005capacity}
A.~F. Dana and B.~Hassibi, ``The capacity region of multiple input erasure
  broadcast channels,'' in {\em IEEE International Symposium on Information
  Theory (ISIT)}, pp.~2315--2319, IEEE, 2005.

\bibitem{georgiadis2009broadcast}
L.~Georgiadis and L.~Tassiulas, ``Broadcast erasure channel with
  feedback-capacity and algorithms,'' in {\em Workshop on Network Coding,
  Theory, and Applications (NetCod'09)}, pp.~54--61, IEEE, 2009.

\bibitem{sundararajan2008arq}
J.~K. Sundararajan, D.~Shah, and M.~M{\'e}dard, ``{ARQ} for network coding,''
  in {\em IEEE International Symposium on Information Theory (ISIT)},
  pp.~1651--1655, IEEE, 2008.

\bibitem{wang2012capacity}
C.-C. Wang, ``On the capacity of 1-to-$ k $ broadcast packet erasure channels
  with channel output feedback,'' {\em IEEE Transactions on Information
  Theory}, vol.~58, no.~2, pp.~931--956, 2012.

\bibitem{linbursty}
S.-C. Lin and I.-H. Wang, ``Bursty broadcast channels with hybrid csit.''

\bibitem{maddah2012completely}
M.~A. Maddah-Ali and D.~Tse, ``Completely stale transmitter channel state
  information is still very useful,'' {\em IEEE Transactions on Information
  Theory}, vol.~58, no.~7, pp.~4418--4431, 2012.

\bibitem{vahid2016approximate}
A.~Vahid, M.~A. Maddah-Ali, and A.~S. Avestimehr, ``Approximate capacity region
  of the {MISO} broadcast channels with delayed {CSIT},'' {\em IEEE
  Transactions on Communications}, vol.~64, no.~7, pp.~2913--2924, 2016.

\bibitem{abdoli2013degrees}
M.~J. Abdoli, A.~Ghasemi, and A.~K. Khandani, ``On the degrees of freedom of
  {K}-user {SISO} interference and {X} channels with delayed {CSIT},'' {\em
  IEEE transactions on Information Theory}, vol.~59, no.~10, pp.~6542--6561,
  2013.

\bibitem{GhasemiX1}
A.~Ghasemi, A.~S. Motahari, and A.~K. Khandani, ``On the degrees of freedom of
  {X} channel with delayed {CSIT},'' in {\em IEEE International Symposium on
  Information Theory Proceedings (ISIT)}, pp.~767--770, IEEE, 2011.

\bibitem{vaze2012degrees}
C.~S. Vaze and M.~K. Varanasi, ``The degrees of freedom region and interference
  alignment for the {MIMO} interference channel with delayed {CSIT},'' {\em
  IEEE Transactions on Information Theory}, vol.~58, no.~7, pp.~4396--4417,
  2012.

\bibitem{Jafar_Retrospective}
H.~Maleki, S.~A. Jafar, and S.~Shamai, ``Retrospective interference
  alignment,'' in {\em IEEE International Symposium on Information Theory
  Proceedings (ISIT)}, pp.~2756--2760, 2011.

\bibitem{tandon2013degrees}
R.~Tandon, S.~Mohajer, H.~V. Poor, and S.~Shamai, ``Degrees of freedom region
  of the {MIMO} interference channel with output feedback and delayed {CSIT},''
  {\em IEEE Transactions on Information Theory}, vol.~59, no.~3,
  pp.~1444--1457, 2013.

\bibitem{nazer2012ergodic}
B.~Nazer, M.~Gastpar, S.~A. Jafar, and S.~Vishwanath, ``Ergodic interference
  alignment,'' {\em IEEE Transactions on Information Theory}, vol.~58, no.~10,
  pp.~6355--6371, 2012.

\bibitem{lashgari2014linear}
S.~Lashgari, A.~S. Avestimehr, and C.~Suh, ``Linear degrees of freedom of the
  {X}-channel with delayed {CSIT},'' {\em IEEE Transactions on Information
  Theory}, vol.~60, no.~4, pp.~2180--2189, 2014.

\bibitem{abdoli2014layered}
M.~J. Abdoli and A.~S. Avestimehr, ``Layered interference networks with delayed
  {CSI}: {D}o{F} scaling with distributed transmitters,'' {\em IEEE
  Transactions on Information Theory}, vol.~60, no.~3, pp.~1822--1839, 2014.

\bibitem{vahid2015informational}
A.~Vahid, I.~Shomorony, and R.~Calderbank, ``Informational bottlenecks in
  two-unicast wireless networks with delayed {CSIT},'' in {\em 53rd Annual
  Allerton Conference on Communication, Control, and Computing (Allerton)},
  pp.~1256--1263, IEEE, 2015.

\bibitem{vahidXISIT}
A.~Vahid, ``Finite field {X}-channels with delayed {CSIT} and common
  messages,'' {\em to appear in the IEEE International Symposium on Information
  Theory (ISIT)}, 2018.

\bibitem{kobayashi2012degrees}
M.~Kobayashi, S.~Yang, D.~Gesbert, and X.~Yi, ``On the degrees of freedom of
  time correlated {MISO} broadcast channel with delayed {CSIT},'' in {\em IEEE
  International Symposium on Information Theory Proceedings (ISIT)},
  pp.~2501--2505, IEEE, 2012.

\bibitem{yang2013degrees}
S.~Yang, M.~Kobayashi, D.~Gesbert, and X.~Yi, ``Degrees of freedom of time
  correlated {MISO} broadcast channel with delayed {CSIT},'' {\em IEEE
  Transactions on Information Theory}, vol.~59, no.~1, pp.~315--328, 2013.

\bibitem{de2016optimal}
P.~de~Kerret, D.~Gesbert, J.~Zhang, and P.~Elia, ``Optimal sum-{D}of of the
  {K}-user {MISO} {BC} with current and delayed feedback,'' {\em arXiv preprint
  arXiv:1604.01653}, 2016.

\bibitem{gou2012optimal}
T.~Gou and S.~A. Jafar, ``Optimal use of current and outdated channel state
  information: {D}egrees of freedom of the {MISO} {BC} with mixed {CSIT},''
  {\em IEEE Communications Letters}, vol.~16, no.~7, pp.~1084--1087, 2012.

\bibitem{yi2014degrees}
X.~Yi, S.~Yang, D.~Gesbert, and M.~Kobayashi, ``The degrees of freedom region
  of temporally correlated {MIMO} networks with delayed {CSIT},'' {\em IEEE
  Transactions on Information Theory}, vol.~60, no.~1, pp.~494--514, 2014.

\bibitem{chen2014symmetric}
J.~Chen and P.~Elia, ``Symmetric two-user {MIMO} {BC} with evolving feedback,''
  in {\em Information Theory and Applications Workshop (ITA)}, pp.~1--5, IEEE,
  2014.

\bibitem{gesbert2002outdoor}
D.~Gesbert, H.~Bolcskei, D.~A. Gore, and A.~J. Paulraj, ``Outdoor {MIMO}
  wireless channels: {M}odels and performance prediction,'' {\em IEEE
  Transactions on Communications}, vol.~50, no.~12, pp.~1926--1934, 2002.

\bibitem{goldsmith2003capacity}
A.~Goldsmith, S.~A. Jafar, N.~Jindal, and S.~Vishwanath, ``Capacity limits of
  {MIMO} channels,'' {\em IEEE Journal on selected areas in Communications},
  vol.~21, no.~5, pp.~684--702, 2003.

\bibitem{love2008overview}
D.~J. Love, R.~W. Heath, V.~K. Lau, D.~Gesbert, B.~D. Rao, and M.~Andrews, ``An
  overview of limited feedback in wireless communication systems,'' {\em IEEE
  Journal on selected areas in Communications}, vol.~26, no.~8, 2008.

\bibitem{raghavan2013statistical}
V.~Raghavan, S.~V. Hanly, and V.~V. Veeravalli, ``Statistical beamforming on
  the grassmann manifold for the two-user broadcast channel,'' {\em IEEE
  Transactions on Information Theory}, vol.~59, no.~10, pp.~6464--6489, 2013.

\bibitem{choi2013limited}
J.~Choi, V.~Raghavan, and D.~J. Love, ``Limited feedback design for the
  spatially correlated multi-antenna broadcast channel,'' in {\em IEEE Global
  Communications Conference (GLOBECOM)}, pp.~3481--3486, IEEE, 2013.

\bibitem{dai2015transmit}
M.~Dai and B.~Clerckx, ``Transmit beamforming for {MISO} broadcast channels
  with statistical and delayed {CSIT},'' {\em IEEE Transactions on
  Communications}, vol.~63, no.~4, pp.~1202--1215, 2015.

\bibitem{clerckx2015space}
B.~Clerckx and D.~Gesbert, ``Space-{T}ime encoded {MISO} broadcast channel with
  outdated {CSIT}: {A}n error rate and diversity performance analysis,'' {\em
  IEEE Transactions on Communications}, vol.~63, no.~5, pp.~1661--1675, 2015.

\bibitem{zhang2017spatially}
F.~Zhang, M.~Fadel, and A.~Nosratinia, ``Spatially correlated {MIMO} broadcast
  channel: {A}nalysis of overlapping correlation eigenspaces,'' in {\em IEEE
  International Symposium on Information Theory (ISIT)}, pp.~1097--1101, IEEE,
  2017.

\bibitem{nosrat2011mimo}
B.~Nosrat-Makouei, J.~G. Andrews, and R.~W. Heath, ``{MIMO} interference
  alignment over correlated channels with imperfect {CSI},'' {\em IEEE
  Transactions on Signal Processing}, vol.~59, no.~6, pp.~2783--2794, 2011.

\bibitem{raghavan2010statistical}
V.~Raghavan and S.~Hanly, ``Statistical beamformer design for the two-antenna
  interference channel,'' in {\em IEEE International Symposium on Information
  Theory Proceedings (ISIT)}, pp.~2278--2282, IEEE, 2010.

\bibitem{filippou2013optimal}
M.~C. Filippou, D.~Gesbert, and G.~A. Ropokis, ``Optimal combining of
  instantaneous and statistical {CSI} in the {SIMO} interference channel,'' in
  {\em IEEE Vehicular Technology Conference (VTC Spring)}, pp.~1--5, IEEE,
  2013.

\bibitem{duan2016state}
R.~Duan, Y.~Liang, and S.~S. Shitz, ``State-dependent gaussian interference
  channels: Can state be fully canceled?,'' {\em IEEE Transactions on
  Information Theory}, vol.~62, no.~4, pp.~1957--1970, 2016.

\bibitem{duan2013state}
R.~Duan, Y.~Liang, A.~Khisti, and S.~S. Shitz, ``State-dependent gaussian
  z-channel with mismatched side-information and interference,'' in {\em
  Information theory workshop (ITW), 2013 IEEE}, pp.~1--5, IEEE, 2013.

\bibitem{ghasemi2014state}
S.~Ghasemi-Goojani and H.~Behroozi, ``State-dependent gaussian z-interference
  channel: New results,'' in {\em Information Theory and its Applications
  (ISITA), 2014 International Symposium on}, pp.~468--472, IEEE, 2014.

\bibitem{wainwright2008graphical}
M.~J. Wainwright, M.~I. Jordan, {\em et~al.}, ``Graphical models, exponential
  families, and variational inference,'' {\em Foundations and
  Trends{\textregistered} in Machine Learning}, vol.~1, no.~1--2, pp.~1--305,
  2008.

\bibitem{kao2017linear}
D.~T. Kao and A.~S. Avestimehr, ``Linear degrees of freedom of the {MIMO}
  {X}-channel with delayed {CSIT},'' {\em IEEE Transactions on Information
  Theory}, vol.~63, no.~1, pp.~297--319, 2017.

\bibitem{yang2017linear}
M.~Yang, S.-W. Jeon, and D.~K. Kim, ``Linear degrees of freedom of {MIMO}
  broadcast channels with reconfigurable antennas in the absence of {CSIT},''
  {\em IEEE Transactions on Information Theory}, vol.~63, no.~1, pp.~320--335,
  2017.

\bibitem{bernstein1924modification}
S.~Bernstein, ``On a modification of chebyshev's inequality and of the error
  formula of laplace,'' {\em Ann. Sci. Inst. Sav. Ukraine, Sect. Math}, vol.~1,
  no.~4, pp.~38--49, 1924.

\bibitem{Chernoff}
J.~Phillips, ``Chernoff-hoeffding inequality and applications,'' {\em arXiv
  preprint arXiv:1209.6396}, 2012.

\bibitem{Hoeffding}
W.~Hoeffding, ``Probability inequalities for sums of bounded random
  variables,'' {\em Journal of the American Statistical Association}, vol.~58,
  no.~301, pp.~13--30, 1963.

\bibitem{sun2013topological}
H.~Sun, C.~Geng, and S.~A. Jafar, ``Topological interference management with
  alternating connectivity,'' in {\em Information Theory Proceedings (ISIT),
  2013 IEEE International Symposium on}, pp.~399--403, IEEE, 2013.

\bibitem{lashgari2014blind}
S.~Lashgari and A.~S. Avestimehr, ``Blind wiretap channel with delayed
  {CSIT},'' in {\em IEEE International Symposium on Information Theory (ISIT)},
  pp.~36--40, IEEE, 2014.

\bibitem{mukherjee2017secure}
P.~Mukherjee, R.~Tandon, and S.~Ulukus, ``Secure degrees of freedom region of
  the two-user {MISO} broadcast channel with alternating {CSIT},'' {\em IEEE
  Transactions on Information Theory}, 2017.

\bibitem{boutros2002iterative}
J.~Boutros and G.~Caire, ``Iterative multiuser joint decoding: {U}nified
  framework and asymptotic analysis,'' {\em IEEE Transactions on Information
  Theory}, vol.~48, no.~7, pp.~1772--1793, 2002.

\bibitem{verdu1998multiuser}
S.~Verdu, {\em Multiuser detection}.
\newblock Cambridge university press, 1998.

\bibitem{reynolds2002blind}
D.~Reynolds, X.~Wang, and H.~V. Poor, ``Blind adaptive space-time multiuser
  detection with multiple transmitter and receiver antennas,'' {\em IEEE
  Transactions on Signal Processing}, vol.~50, no.~6, pp.~1261--1276, 2002.

\bibitem{el2000iterative}
H.~El~Gamal and E.~Geraniotis, ``Iterative multiuser detection for coded {CDMA}
  signals in {AWGN} and fading channels,'' {\em IEEE Journal on Selected Areas
  in Communications}, vol.~18, no.~1, pp.~30--41, 2000.

\bibitem{poor2004iterative}
H.~V. Poor, ``Iterative multiuser detection,'' {\em Signal Processing
  Magazine}, vol.~21, no.~1, pp.~81--88, 2004.

\bibitem{gollakota2008zigzag}
S.~Gollakota and D.~Katabi, {\em Zigzag decoding: combating hidden terminals in
  wireless networks}, vol.~38.
\newblock ACM, 2008.

\bibitem{pan2016stability}
H.~Pan, J.~Li, and K.~Wang, ``On stability region of two-user interference
  network with permission to utilize past receptions,'' {\em IEEE Access},
  2016.

\end{thebibliography}

\begin{IEEEbiography} [{\includegraphics[width=25mm]{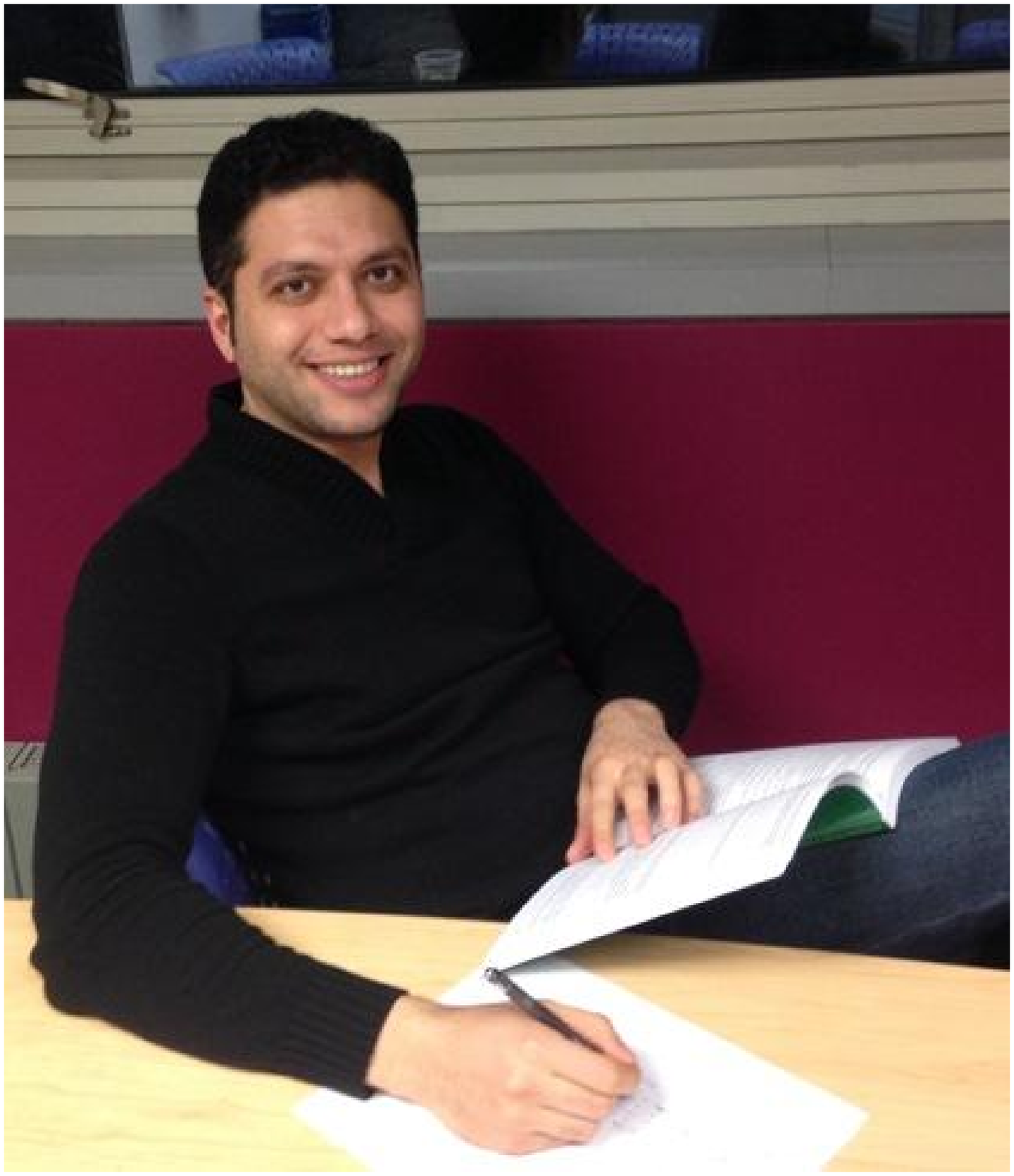}}] {Alireza Vahid} received the B.Sc. degree in electrical engineering from Sharif University of Technology, Tehran, Iran, in 2009, and the M.Sc. degree and Ph.D. degree in electrical and computer engineering both from Cornell University, Ithaca, NY, in 2012 and 2015 respectively. From 2015 to 2017 he worked as a postdoctoral scholar at the Information Initiative at Duke University, Durham, NC. He is currently an assistant professor of electrical engineering at the University of Colorado, Denver. His research interests include network information theory, wireless communications, statistics and machine learning.

Dr. Vahid received the 2015 Outstanding PhD Thesis Research Award at Cornell University. He also received the Director's Ph.D. Teaching Assistant Award in 2010, Jacobs Scholar Fellowship in 2009, and Qualcomm Innovation Fellowship in 2013.
\end{IEEEbiography}

\begin{IEEEbiography} [{\includegraphics[height=32mm]{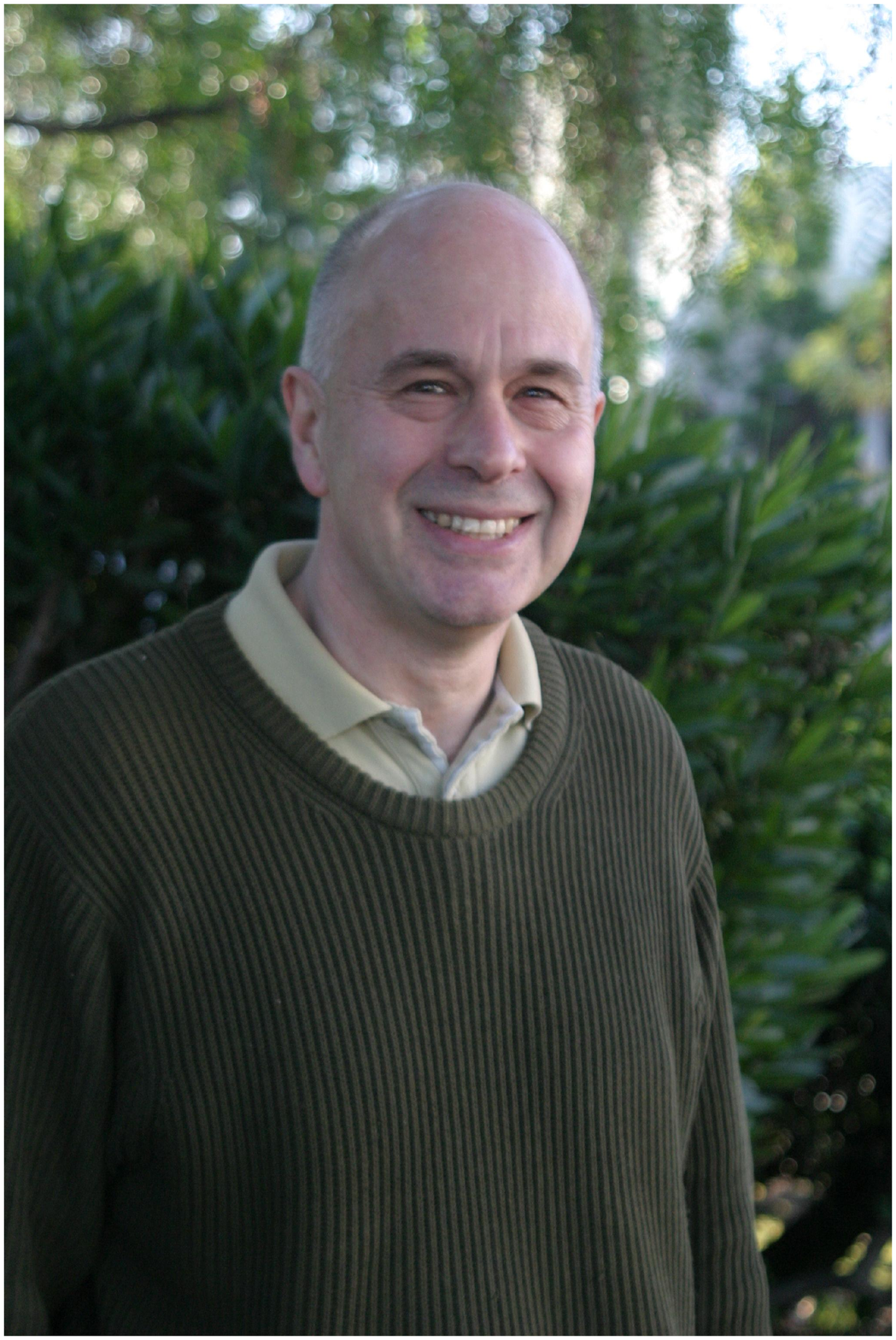}}] {Robert Calderbank} (M’89 – SM’97 – F’98) received the BSc degree in 1975 from Warwick University, England, the MSc degree in 1976 from Oxford University, England, and the PhD degree in 1980 from the California Institute of Technology, all in mathematics.

Dr. Calderbank is Professor of Electrical Engineering at Duke University where he now directs the Information Initiative at Duke (\emph{i}iD) after serving as Dean of Natural Sciences (2010-2013). Dr. Calderbank was previously Professor of Electrical Engineering and Mathematics at Princeton University where he directed the Program in Applied and Computational Mathematics. Prior to joining Princeton in 2004, he was Vice President for Research at AT\&T, responsible for directing the first industrial research lab in the world where the primary focus is data at scale. At the start of his career at Bell Labs, innovations by Dr. Calderbank were incorporated in a progression of voiceband modem standards that moved communications practice close to the Shannon limit. Together with Peter Shor and colleagues at AT\&T Labs he showed that good quantum error correcting codes exist and developed the group theoretic framework for quantum error correction. He is a co-inventor of space-time codes for wireless communication, where correlation of signals across different transmit antennas is the key to reliable transmission.

Dr. Calderbank served as Editor in Chief of the IEEE TRANSACTIONS ON INFORMATION THEORY from 1995 to 1998, and as Associate Editor for Coding Techniques from 1986 to 1989. He was a member of the Board of Governors of the IEEE Information Theory Society from 1991 to 1996 and from 2006 to 2008. Dr. Calderbank was honored by the IEEE Information Theory Prize Paper Award in 1995 for his work on the $Z_4$ linearity of Kerdock and Preparata Codes (joint with A.R. Hammons Jr., P.V. Kumar, N.J.A. Sloane, and P. Sole), and again in 1999 for the invention of space-time codes (joint with V. Tarokh and N. Seshadri). He has received the 2006 IEEE Donald G. Fink Prize Paper Award, the IEEE Millennium Medal, the 2013 IEEE Richard W. Hamming Medal, the 2015 Shannon Award, and he was elected to the US National Academy of Engineering in 2005.

\end{IEEEbiography}

\end{document}